\documentclass[a4paper, twoside, 11pt, reqno]{amsart}

\usepackage[T1]{fontenc}
\usepackage[english]{babel}
\usepackage[a4paper, hmargin=2.75cm, vmargin=3cm]{geometry}
\usepackage[foot]{amsaddr}
\usepackage{lmodern, microtype, setspace, csquotes, enumitem, moreenum, tikz, comment}
\usepackage{amsmath, amsfonts, amssymb, amsthm, mathtools, mathrsfs, bbm}
\usepackage[colorlinks=true, allcolors=blue!50!black, breaklinks=true]{hyperref}
\usepackage[capitalize]{cleveref}

\setdisplayskipstretch{2}

\setlist{labelindent=\parindent, leftmargin=*, align=left, topsep=8pt, itemsep=4pt, parsep=0pt}
\setlist[enumerate]{label=\normalfont(\alph*)}
\SetEnumitemKey{env}{wide=\parindent, label=\normalfont\bfseries(\arabic*), ref=\normalfont(\arabic*), nosep}
\SetEnumitemKey{equiv}{label=\normalfont(\roman*), align=right}

\theoremstyle{definition}
\newtheorem{definition}{Definition}[section]
\newtheorem{assumption}[definition]{Assumption}
\newtheorem{remark}[definition]{Remark}
\newtheorem{example}[definition]{Example}
\newtheorem{examples}[definition]{Examples}
\theoremstyle{plain}
\newtheorem{theorem}[definition]{Theorem}
\newtheorem{proposition}[definition]{Proposition}
\newtheorem{lemma}[definition]{Lemma}
\newtheorem{corollary}[definition]{Corollary}
\newtheorem*{theorem*}{Theorem}
\AddToHook{env/assumption/begin}{\crefalias{definition}{assumption}}
\AddToHook{env/remark/begin}{\crefalias{definition}{remark}}
\AddToHook{env/example/begin}{\crefalias{definition}{example}}
\AddToHook{env/examples/begin}{\crefalias{definition}{examples}}
\AddToHook{env/theorem/begin}{\crefalias{definition}{theorem}}
\AddToHook{env/proposition/begin}{\crefalias{definition}{proposition}}
\AddToHook{env/lemma/begin}{\crefalias{definition}{lemma}}
\AddToHook{env/corollary/begin}{\crefalias{definition}{corollary}}
\crefname{examples}{Example}{Examples}

\renewcommand{\implies}{\Rightarrow}
\renewcommand{\subset}{\subseteq}
\renewcommand{\supset}{\supseteq}
\renewcommand{\Re}{\operatorname{Re}}
\renewcommand{\Im}{\operatorname{Im}}
\let\originalleft\left
\let\originalright\right
\renewcommand{\left}{\mathopen{}\mathclose\bgroup\originalleft}
\renewcommand{\right}{\aftergroup\egroup\originalright}
\renewcommand{\Delta}{\varDelta}
\renewcommand{\Xi}{\varXi}
\renewcommand{\Sigma}{\varSigma}
\renewcommand{\Phi}{\varPhi}
\renewcommand{\Psi}{\varPsi}
\renewcommand{\Omega}{\varOmega}
\newcommand{\N}{\mathbb{N}}
\newcommand{\Z}{\mathbb{Z}}
\newcommand{\R}{\mathbb{R}}
\newcommand{\C}{\mathbb{C}}
\newcommand{\MD}{\mathcal{D}}
\newcommand{\MF}{\mathcal{F}}
\newcommand{\MH}{\mathcal{H}}
\newcommand{\MK}{\mathcal{K}}
\newcommand{\MS}{\mathcal{S}}
\newcommand{\MFM}{\mathfrak{M}}
\newcommand{\MFS}{\mathfrak{S}}
\newcommand{\ii}{\mathrm{i}}
\newcommand{\ee}{\mathrm{e}}
\newcommand{\BO}{\mathscr{B}}
\newcommand{\NO}{\BO_1}
\newcommand{\HS}{\BO_2}
\newcommand{\UO}{\mathscr{U}}
\newcommand{\comm}[1]{#1^{\prime}}
\newcommand{\NF}[1]{#1_{\ast}^{+}}
\newcommand{\NS}[1]{\Sigma_{\ast}(#1)}
\newcommand{\NPC}{\mathcal{P}}
\renewcommand{\AE}[1]{#1^{(\eta)}}
\newcommand{\id}{\mathrm{Id}}
\newcommand{\1}{\mathbf{1}}
\newcommand{\diff}{\mathop{}\!\mathrm{d}}
\newcommand{\dlim}[2]{\operatorname*{\text{$#1$}-lim}_{#2}\,}
\newcommand{\ssupp}{\mathbf{s}}
\newcommand{\EV}{\mathbb{E}}

\DeclareMathOperator{\dom}{dom}
\DeclareMathOperator{\ran}{ran}
\DeclareMathOperator{\supp}{supp}
\DeclareMathOperator{\lin}{lin}
\DeclareMathOperator{\Aut}{Aut}
\DeclareMathOperator{\tr}{tr}
\newcommand{\set}[2][]{#1\{{#2}#1\}}
\newcommand{\Set}[1]{\left\{#1\right\}}
\newcommand{\abs}[2][]{#1\vert{#2}#1\vert}

\newcommand{\norm}[2][]{#1\Vert{#2}#1\Vert}

\newcommand{\braket}[2][]{#1\langle{#2}#1\rangle}
\newcommand{\Braket}[1]{\left\langle #1 \right\rangle}
\newcommand{\odbound}[4][]{\frac{\diff^{#1} #2}{\diff {#3}^{#1}}_{\upharpoonright \, #4}}
\newcommand{\restr}{\!\upharpoonright\!}
\newcommand{\ce}{\coloneqq}

\newcommand{\oln}[1]{\overline{#1}}
\newcommand{\wt}[1]{\widetilde{#1}}
\newcommand{\inv}[1]{#1^{-1}}
\newcommand{\tand}{\quad\text{and}\quad}
\newcommand{\com}{{\, , \quad}}
\newcommand{\tAd}{\textit{Ad}}
\newcommand{\ndot}{\norm{\,\cdot\,}}
\newcommand{\bdot}{\braket{\cdot, \cdot}}
\newcommand{\mop}{\mathrm{op}}
\newcommand{\ph}{\varphi}
\newcommand{\Gc}{G_{c}}

\begin{document}

\title[The two-sided Bogoliubov inequality in von Neumann algebras]{The two-sided Bogoliubov inequality in von Neumann algebras conceptualizes the free energy--quantum correlations link}

\author[B. M. Reible]{Benedikt M. Reible$^{1,2,3}$}
\author[A. Much]{Albert Much$^{3}$}
\author[R. Verch]{Rainer Verch$^{3}$}
\author[C. Schütte]{Christof Schütte$^{1,2}$}
\author[L. Delle Site]{Luigi Delle Site$^2$}

\address{$^1$Zuse Institute Berlin, Takustr.~7, 14195 Berlin, Germany}
\address{$^2$Institute of Mathematics, Freie Universität Berlin, Arnimallee 6, 14195 Berlin, Germany}
\address{$^3$Institute of Theoretical Physics, Universität Leipzig, Brüderstr.~16, 04103 Leipzig, Germany}

\email{benedikt.reible@fu-berlin.de}
\email{much@itp.uni-leipzig.de} 
\email{rainer.verch@uni-leipzig.de}
\email{Christof.Schuette@fu-berlin.de}
\email{luigi.dellesite@fu-berlin.de}


\begin{abstract}
    The quantum-mechanical two-sided Bogoliubov inequality provides upper and lower bounds for the free energy required to separate a system of interacting particles into independent subsystems. The bounds can be calculated straightforwardly from the ensemble average of the interface energy, bypassing the direct evaluation of the free energy. In this work, we generalize the two-sided Bogoliubov inequality to arbitrary von Neumann algebras by employing the Araki-Uhlmann relative entropy and the framework of unbounded perturbation theory of KMS states. Furthermore, we obtain variational expressions for the relative free energy that extend existing bounded-perturbation principles to the unbounded setting. Crucially, these mathematical developments yield a physically well-founded thermodynamic criterion for the quantification of entanglement in infinite-dimensional systems.
\end{abstract}

\maketitle

\section{Introduction}

\subsection{Quantifying correlations}

In quantum systems consisting of infinitely many degrees of freedom, states across two subsystems can be infinitely entangled, and the theory of von Neumann algebras provides a powerful mathematical framework in which these correlations can be formalized \cite{Keyl2003, vanLuijk2025cmp, Verch2005}. Natural realizations of such systems are given by quantum field theoretic models of interacting particles, for example, quantum fields of spins serving as realistic prototypes for modern quantum simulators. The future of quantum computation, and more generally of quantum technology, requires a detailed understanding of these systems, in particular, the ability of practitioners to quantify entanglement and transfer this information onto optimally designed devices \cite{Amico2008, Kinoshita2025}.

Common entanglement measures, which quantify the difference between a reference state and the separable states of the system, are the relative entropy of entanglement, the entanglement cost, and the distillable entanglement, to name a few \cite{Horodecki2009, Plenio2007, Vedral1998, Vedral1997}; they are amenable to calculations or estimations using state-of-the-art computational protocols \cite{Feng2024, Koutny2023}. In recent years, considerable progress has been made in the mathematical analysis of the role and properties of relative entropy-like quantities \cite{Hiai2021mps, Rubboli2024, Vedral2002}. While the perspective of physical applications is rather promising, complementary paths that avoid direct calculations of the relative entropy, if physically well founded and mathematically rigorous, can lead to other possible experimental techniques (in particular, thermodynamic ones, as discussed below) and are certainly of major need in this field, which is in continuous development; our current contribution must be located in this perspective. To circumvent the challenges related to the direct determination of these entropic quantities, one can exploit links between information theory and statistical mechanics.

We propose a free energy-like criterion to quantify quantum entanglement by establishing a generalization of the quantum-mechanical \emph{two-sided Bogoliubov inequality} \cite{DelleSite2017, Reible2022} within the mathematical framework of von Neumann algebras; compared to the previously studied case, this generalized inequality offers an analytic tool for an extended range of physical situations, namely, systems with infinitely many degrees of freedom. The inequality provides upper and lower bounds for the free energy cost required for separating a system of interacting particles into two or more non-interacting subsystems. For the prototypical bipartite quantum systems, this leads to a free energy estimate of the correlations between particles located in different subsystems, which can hence be measured thermodynamically; in this context, it should be noted that formal correspondences between thermodynamics and entanglement have been known for quite some time \cite{Popescu1997}. In the past, the two-sided Bogoliubov inequality has been applied to low-temperature dense quantum gases with finite-range interactions \cite{DelleSite2024pra} and to a gas of interacting electrons \cite{Reible2023}, and it has been used to quantify finite-size effects in molecular simulations \cite{Reible2025apx, Reible2025camcos}. The connection between electron correlations and quantum entanglement \cite{DelleSite2015} justifies the idea that the inequality may also be employed to quantify the entanglement between separated subsystems, with the free energy of separation serving as a direct thermodynamic estimate of it. Indeed, the difference between the free energy of the entire system and of the system divided in two identical but separated, non-interacting parts expresses the difference in correlations across the two parts. Moreover, from the point of view of experimental physics, where entanglement is usually measured through highly complex techniques involving delicate preparation and careful manipulation of quantum states \cite{Laurell2025, Xie2023}, thermodynamic (calorimetric) measurements complementing the standard techniques are already being employed in the design of modern quantum materials \cite{Cole2025, Li2019}; while these measurements are not yet used to detect quantum entanglement directly, a possible path for achieving this has been proposed very recently \cite{Bakshi2024, Stamatova2025}.

Connected with the quantum-mechanical two-sided Bogoliubov inequality, there are also variational bounds for the free energy of separation \cite{Reible2022}; they will be generalized to the setting of von Neumann algebras as well. This makes it possible, in principle, to optimize the bounds with respect to some key physical quantities of the system, namely, states and perturbations. In applications to materials modeling, an optimization procedure of this kind could eventually lead to an optimal design of a physical system with prescribed quantum characteristics.

In conclusion, in this paper we provide the conceptual justification, with its proper mathematical structure, for the possibility of measuring entanglement through thermodynamics, while we leave concrete applications of this criterion to systems of interest for future work.

\subsection{Main results and outline of the paper}

The physical ideas sketched above rely on a mathematical development---the extension of the two-sided Bogoliubov inequality---which forms the central part of our proposal as it provides the basis for their solidity.

To illustrate our point of departure, we will first review the quantum-mechanical two-sided Bogoliubov inequality in \cref{sec:BogoliubovQM}. Its desired extension to general von Neumann algebras essentially rests on two mathematical pillars: the Araki-Uhlmann relative entropy \cite{Araki1976, Araki1977, Uhlmann1977} and the framework of unbounded perturbation theory of KMS states developed by J.~Derezi\'{n}ski, V.~Jak\v{s}i\'{c}, and C.-A.~Pillet \cite{DJP03}. Consequently, to set the stage for the extension, we will recall the necessary background from modular theory in \cref{sec:vonNeumann} and present the key ideas of $W^\ast$-dynamical perturbation theory, the results of Ref.~\cite{DJP03}, as well as some additions to these results in \cref{sec:perturbationTheory}.

Our first main result, the two-sided Bogoliubov inequality in von Neumann algebras, will be established in \cref{sec:extension}. Leaving aside some details, in its essential form it can be stated as follows:

\begin{theorem*}[Informal; see \cref{thm:BogoliubovVN}]
    Let $\MFM \subset \BO(\MH)$ be a von Neumann algebra on a Hilbert space $\MH$, let $\omega$ be a faithful KMS-state on $\MFM$, let $U$ be a self-adjoint operator affiliated with $\MFM$ satisfying certain properties, and let $\omega_U$ be the perturbed state. Then
    \begin{equation}\label{eq:informalBogoliubovVN}
        \omega_U(U) \le \MF(\omega_U, \omega) \le \omega(U) \, .
    \end{equation}
\end{theorem*}

$\MF(\omega_U, \omega)$ denotes the relative free energy between the perturbed and unperturbed state, which we define in analogy to the quantum-mechanical case (cf. \cref{def:relativeFreeEnergy}). As indicated above, the proof of the theorem uses the framework of Derezi\'{n}ski, Jak\v{s}i\'{c}, and Pillet \cite{DJP03} as well as some extensions of the latter (in particular, \cref{lem:unboundedPerturbationRelEnt2}). To make the connection of this theorem to our previous work \cite{Reible2022} explicit, we will show that if $\MFM = \BO(\MH)$ is the von Neumann algebra of all bounded operators on $\MH$, used to describe non-relativistic quantum mechanics of finitely many degrees of freedom, then \cref{eq:informalBogoliubovVN} reduces to the quantum-mechanical two-sided Bogoliubov inequality (cf. \cref{pro:BogoliubovVNimpliesQM}).

While the inequality \eqref{eq:informalBogoliubovVN} establishes rigorous, computable bounds for the relative free energy, sharpening them in terms of variational principles provides a possible optimization framework for practical materials design. Our second main result addresses this by establishing variational bounds for $\MF(\omega_U, \omega)$, which can be summarized informally as follows:

\begin{theorem*}[Informal; see \cref{thm:variationalBoundsVN}]
    Let $\MFM \subset \BO(\MH)$ be a von Neumann algebra, let $\omega$ be a faithful KMS-state on $\MFM$, let $U$ be a self-adjoint operator affiliated with $\MFM$ satisfying certain properties, and let $\omega_U$ be the perturbed state. Then
    \begin{equation}\label{eq:informalVariationalBoundsVN}
        \inf_{\psi} \Bigl(\psi(U) + \inv{\beta} \MS_\MFM(\psi, \omega)\Bigr) = \MF(\omega_U, \omega) = \sup_{V} \Bigl(\omega_U(U - V) - \inv{\beta} \log \bigl(\norm{\Omega_{- \beta V}}^2\bigr)\Bigr) \, ,
    \end{equation}
    where the infimum is taken over a suitable subset of normal states on $\MFM$, and the supremum is taken over a certain subset of self-adjoint operators affiliated with $\MFM$.
\end{theorem*}

In the above expression, $\MS_\MFM(\psi, \omega)$ denotes the Araki-Uhlmann relative entropy and $\Omega_{- \beta V}$ the perturbation of the vector $\Omega \in \MH$ with the operator $V$ (cf. \cref{def:unboundedPerturbedState}). The proof of this theorem relies on variational principles for the quantities $\log \bigl(\norm{\Omega_{- \beta V}}^2\bigr)$ and $\MS_\MFM(\psi, \omega)$ that were established by D.~Petz \cite{Petz1988} for bounded perturbations; we will extend both of them to certain unbounded perturbations first, from which \cref{eq:informalVariationalBoundsVN} then follows.

\subsection{Conventions}

The symbol $\MH$ will always denote a complex Hilbert space with corresponding inner product $\bdot$ (assumed to be linear in the second argument), and $(\BO(\MH), \ndot_\mop)$, $(\NO(\MH), \ndot_\mathrm{tr})$, and $(\HS(\MH), \ndot_\mathrm{HS})$ will denote the Banach spaces of bounded, trace-class, and Hilbert-Schmidt operators on $\MH$, respectively.

In the following, unless stated otherwise, the theorems, propositions, and lemmas for which a proof is provided are original results.

\section{The quantum-mechanical two-sided Bogoliubov inequality}\label{sec:BogoliubovQM}

Let us begin by reviewing the two-sided Bogoliubov inequality in its formulation for non-relativistic quantum systems consisting of finitely many degrees of freedom, which was established in Ref.~\cite{Reible2022}. (See also Ref.~\cite{DelleSite2017} for classical systems.)

\subsection{Setup}\label{subsec:BogoliubovSetup}

Consider a quantum system at inverse temperature $\beta > 0$ confined to a spatial region $\Sigma \subset \R^d$ with lower semi-bounded, self-adjoint Hamiltonian $H : \MH \supset \dom(H) \to \MH$. Suppose that $\ee^{- \beta H} \in \NO(\MH)$, and that the system is described by the \emph{canonical Gibbs state}
\begin{equation}\label{eq:canonicalGibbsState}
    \rho = \frac{1}{Z} \, \ee^{- \beta H} \com Z = \tr(\ee^{- \beta H}) \, .
\end{equation}
Divide $\Sigma$ into $n \in \N$ disjoint subregions $\Sigma_k \subset \R^d$, $k \in \set{1, \dotsc, n}$, such that $\Sigma = \bigcup_{k=1}^{n} \Sigma_k$, and let $H_0^{(k)}$ be the part of the Hamiltonian $H$ acting only on the particles inside $\Sigma_k$. Define $H_0 = \sum_{k=1}^{n} H_0^{(k)}$ and suppose that there exists a self-adjoint operator $U : \MH \supset \dom(U) \to \MH$, mediating the interaction between the subregions $\Sigma_k$, such that $H$ can be written as
\begin{equation*}
    H = H_0 + U \, .
\end{equation*}
Assume that also $\ee^{- \beta H_0}$ is a trace-class operator and define the \enquote{free} Gibbs state
\begin{equation}\label{eq:freeGibbsState}
    \rho_0 = \frac{1}{Z_0} \, \ee^{-\beta H_0} \com Z_0 = \tr(\ee^{- \beta H_0}) \, ,
\end{equation}
representing the thermal state of the uncoupled subsystems. To quantify the thermodynamic difference between the full system and the collection of uncoupled systems, we define the \emph{relative free energy}
\begin{equation}\label{eq:interfaceEnergy}
    \Delta F \ce - \inv{\beta} \log\left(\frac{Z}{Z_0}\right) \, ,
\end{equation}
which is simply the difference between the free energy of the interacting and the non-interacting system. Physically, therefore, $\Delta F$ represents the free energy difference associated with the partitioning of the large system into smaller, independent subsystems.

\subsection{Results}

The two-sided Bogoliubov inequality \cite[Thm. 4.1]{Reible2022} provides upper and lower bounds for $\Delta F$ in terms of expectation values of the operator $U$.

\begin{theorem}[Two-sided Bogoliubov inequality]\label{thm:BogoliubovQM}
    In the situation described above, assume that $\EV_{\rho_0} [U] \ce \tr (\rho_0 U) < + \infty$ and that $\EV_{\rho} [U] \ce \tr (\rho U) < + \infty$. Then
    \begin{equation}\label{eq:BogoliubovQM}
        \EV_{\rho} [U] \le \Delta F \le \EV_{\rho_0} [U] \, .
    \end{equation}
\end{theorem}

\begin{proof}
    Recall the definition of the Umegaki-Lindblad relative entropy \cite{Lindblad73, Umegaki62}: for two density operators $\rho$ and $\sigma$ on $\MH$ with $\ker(\rho)^\perp \subset \ker(\sigma)^\perp$, $S(\rho, \sigma)$ is defined to be \cite[Eq. (1.3)]{OP2004}
    \begin{equation}\label{eq:UmegakiRelEnt}
        S(\rho, \sigma) \ce \tr \bigl(\rho (\log \rho - \log \sigma)\bigr) \, .
    \end{equation}
    Note that from Klein's inequality (see, e.g., Ref.~\cite[Eq. (3.27)]{OP2004}), it follows that $S$ is non-negative: $S(\rho, \sigma) \ge 0$ for all density operators $\rho, \sigma$ on $\MH$.
    
    Since both $\rho$ and $\rho_0$ from \cref{eq:canonicalGibbsState,eq:freeGibbsState} are injective, we can compute $S(\rho, \rho_0)$ and $S(\rho_0, \rho)$ using the previous formula \cite[pp. 8 f.]{Reible2022}:
    \begin{equation}\label{eq:realtiveFreeEnergyFQS}
        S(\rho, \rho_0) = - \beta \, \EV_{\rho}[U] + \beta \Delta F \tand S(\rho_0, \rho) = \beta \, \EV_{\rho_0}[U] - \beta \Delta F \, .
    \end{equation}
    From $S(\rho, \rho_0) \ge 0$ and $S(\rho_0, \rho) \ge 0$, the lower and upper bound for $\Delta F$ follow.
\end{proof}

\begin{remark}\label{rem:BogoliubovQM}
    The upper bound in \cref{eq:BogoliubovQM} is the famous \emph{Peierls-Bogoliubov inequality}, which is an important tool in statistical mechanics. For two self-adjoint operators $A$ and $B$ such that $\ee^A, \ee^{A + B} \in \NO(\MH)$, it is usually written in the form \cite[Cor. 3.15]{OP2004}
    \begin{equation*}
        \log\left(\frac{\tr(\ee^{A + B})}{\tr(\ee^A)}\right) \ge \frac{\tr(B \, \ee^A)}{\tr(\ee^A)} \, .
    \end{equation*}
    Indeed, choosing $A = - \beta H_0$ and $B = - \beta U$, the above inequality becomes $\log(Z / Z_0) \ge - \beta \, \EV_{\rho_0}[U]$ which is equivalent to the upper bound in \cref{eq:BogoliubovQM}.
\end{remark}

\cref{thm:BogoliubovQM} can be sharpened to the following variational bounds, cf. Ref.~\cite[Cor. 4.5]{Reible2022}; the proof relies on the Golden-Thompson inequality \cite[Thm. 1]{BreiteneckerGruemm1972}, which states that for two positive self-adjoint operators $A$ and $B$ on $\MH$ such that $B$ is relatively $A$-bounded, with $A$-bound strictly less than one, and such that $\ee^{-A} \, \ee^{-B} \in \NO(\MH)$, there holds
\begin{equation*}
    \tr\bigl(\ee^{- (A + B)}\bigr) \le \tr(\ee^{-A} \, \ee^{-B}) \, .
\end{equation*}

\begin{theorem}[Variational bounds]\label{thm:variationalBoundsQM}
    Under the assumptions of \cref{thm:BogoliubovQM}, it holds that
    \begin{equation}\label{eq:variationalBoundsQM}
        \inf_{\gamma} \Bigl\{\EV_\gamma[U] + \inv{\beta} S(\gamma, \rho_0)\Bigr\} = \Delta F = \sup_{V \ge 0} \Bigl\{\EV_\rho[U - V] - \inv{\beta} \log \tr\bigl(\ee^{\log \rho_0 - \beta V}\bigr)\Bigr\} \, .
    \end{equation}
    The infimum ranges over all density operators $\gamma$ on $\MH$, and the supremum is taken over all densely defined, positive self-adjoint operators $V$ on $\MH$. In particular, the bounds of \cref{thm:BogoliubovQM} are a special case of \eqref{eq:variationalBoundsQM} obtained by taking $V = 0$ and $\gamma = \rho_0$.
\end{theorem}

\subsection{Example}

The two-sided Bogoliubov inequality can be applied in numerical simulations to estimate finite-size effects \cite{Reible2022, Reible2025apx, Reible2023, Reible2025camcos}. In particular, it can be used to justify the simulation of a small subsystem instead of the computationally unfeasible total system.

To illustrate the underlying principle  \cite[Sect.~5]{Reible2022},  consider a quantum system located in the region $\Sigma \subset \R^3$ composed of $N \in \N$ electrons with positions $x_i \in \R^3$, $i \in \set{1, \dotsc, N}$, and $M \in \N$ nuclei with charges $Z_j \in \Z$ and positions $X_j \in \R^3$, $j \in \set{1, \dotsc, M}$. The electronic Hamiltonian for this problem is given in atomic units by
\begin{equation*}
    H = - \frac{1}{2} \sum_{i=1}^{N} \Delta_{x_i} + \sum_{i=1}^{N} \sum_{j>i}^{N} \frac{1}{\norm{x_i - x_j}} - \sum_{i=1}^{M} \sum_{j=1}^{N} \frac{Z_i}{\norm{X_i - x_j}} \, .
\end{equation*}
Suppose that the system is divided into two smaller subsystems $\Sigma_1, \Sigma_2 \subset \R^3$ such that $\Sigma = \Sigma_1 \sqcup \Sigma_2$; this situation is illustrated in \cref{fig:finiteSizeApproximation}. If, say, $n$ electrons and $m$ nuclei are contained in the region $\Sigma_1$, that is, $x_i, X_j \in \Sigma_1$ for all $1 \le i \le n$ and $1 \le j \le m$, then the Hamiltonian describing just this subsystem takes the form
\begin{equation*}
    H_0^{(1)} = - \frac{1}{2} \sum_{i=1}^{n} \Delta_{x_i} + \sum_{i=1}^{N} \sum_{j>i}^{N} \frac{1}{\norm{x_i - x_j}} - \sum_{i=1}^{m} \sum_{j=1}^{n} \frac{Z_i}{\norm{X_i - x_j}} \, .
\end{equation*}
Similarly, the Hamiltonian describing only the part of the system located in the region $\Sigma_2$, i.e., $x_i, X_j \in \Sigma_2$ for all $1 \le i \le N - n$ and $1 \le j \le M - m$, is given by
\begin{equation*}
    H_0^{(2)} = - \frac{1}{2} \sum_{i=1}^{N-n} \Delta_{x_i} + \sum_{i=1}^{N-n} \sum_{j>i}^{N-n} \frac{1}{\norm{x_i - x_j}} - \sum_{i=1}^{M-m} \sum_{j=1}^{N-n} \frac{Z_i}{\norm{X_i - x_j}} \, .
\end{equation*}
Finally, the operator $U$ mediating the interaction between these two subsystems reads
\begin{equation*}
    U = \sum_{i=1}^{n} \sum_{j=1}^{N-n} \frac{1}{\norm{x_i - x_j}} - \sum_{i=1}^{M-m} \sum_{j=1}^{n} \frac{Z_i}{\norm{X_i - x_j}} - \sum_{i=1}^{m} \sum_{j=1}^{N-n} \frac{Z_i}{\norm{X_i - x_j}} \, .
\end{equation*}

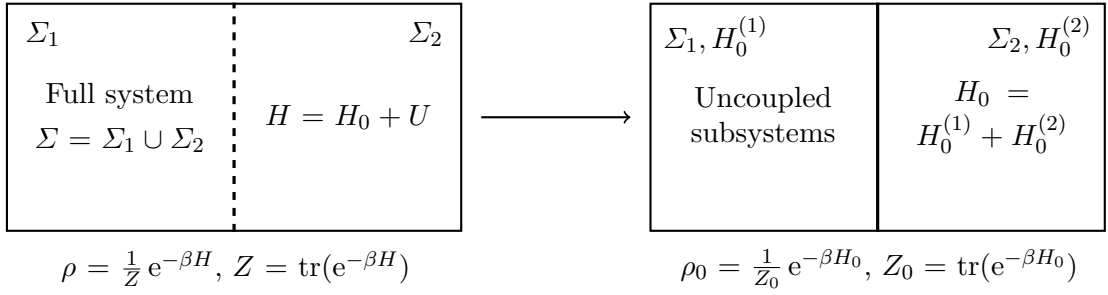
\begin{figure}[t]
    \centering
    \scalebox{1.0}{
    \begin{tikzpicture}[node distance=5cm]
        \node (fullBox) [rectangle, thick, minimum width=6cm, minimum height=3cm, draw, text width=2cm, align=center] {};
        \node [anchor=north west, inner sep=7pt] at (fullBox.north west) {$\Sigma_1$};
        \node [anchor=north east, inner sep=7pt] at (fullBox.north east) {$\Sigma_2$};
        \node [anchor=west, inner sep=7pt, text width=2.5cm, align=center] at (fullBox.west) {Full system \\[4pt] $\Sigma = \Sigma_1 \cup \Sigma_2$};
        \node [anchor=east, inner sep=7pt, text width=2.5cm, align=center] at (fullBox.east) {$H = H_0 + U$};
        \node [text width=5cm, align=center, below of=fullBox, yshift=3cm] {$\rho = \frac{1}{Z} \, \ee^{- \beta H}$, $Z = \tr(\ee^{- \beta H})$};
        \draw[dashed, very thick] (fullBox.south) -- (fullBox.north);

        \node (subBoxes) [rectangle, thick, minimum width=6cm, minimum height=3cm, draw, text width=3cm, align=center, right of=fullBox, xshift=3.5cm] {};
        \node [anchor=north west, inner sep=5pt] at (subBoxes.north west) {$\Sigma_1, H_0^{(1)}$};
        \node [anchor=north east, inner sep=5pt] at (subBoxes.north east) {$\Sigma_2, H_0^{(2)}$};
        \node [anchor=west, inner sep=7pt, text width=2.5cm, align=center] at (subBoxes.west) {Uncoupled subsystems};
        \node [anchor=east, inner sep=7pt, text width=2.5cm, align=center] at (subBoxes.east) {$H_0 =$ \\[2pt] $H_0^{(1)} + H_0^{(2)}$};
        \node [text width=5.5cm, align=center, below of=subBoxes, yshift=3cm] {$\rho_0 = \frac{1}{Z_0} \, \ee^{- \beta H_0}$, $Z_0 = \tr(\ee^{- \beta H_0})$};
        
        \draw[very thick] (subBoxes.south) -- (subBoxes.north);
        \draw[->, thick] ([xshift=0.25cm]fullBox.east) -- ([xshift=-0.25cm]subBoxes.west);
    \end{tikzpicture}}
    \caption{Illustration of the division of a large thermodynamic quantum system into two non-interacting subsystems. Figure adapted from Ref.~\cite{Reible2025apx}.}
    \label{fig:finiteSizeApproximation}
\end{figure}

In this situation, an important question is whether simulating only one of the two subsystems yields reliable numerical results for the local electronic properties of the molecular liquid, and the above illustrates that this problem is equivalent to determining the degree of independence of the two subsystems $\Sigma_1$ and $\Sigma_2$ with respect to the total system $\Sigma$. To solve this problem, one would compute the expectation values $\EV_{\rho_0}[U]$ and $\EV_{\rho}[U]$, where $\rho$ corresponds to the canonical Gibbs state \eqref{eq:canonicalGibbsState} of the full system in the entire region $\Sigma$ and $\rho_0 = \rho_{0}^{(1)} \otimes \rho_{0}^{(2)}$, with $\rho_{0}^{(k)}$ given by \eqref{eq:freeGibbsState} in terms of the Hamiltonian $H_0^{(k)}$, $k \in \set{1, 2}$. Since the interaction $U$ is of one-body and two-body form, the density operator representations needed for the computations are one-body and two-body reduced density matrix terms, i.e., electron densities and two-body electron-electron correlation functions; these quantities are routinely computed in electronic structure calculations.

\section{von Neumann algebras, modular theory, and relative entropy}\label{sec:vonNeumann}

To pave the way for the extension of the two-sided Bogoliubov inequality, we will now briefly recall some background material from the theory of von Neumann algebras, in particular, concerning modular theory and the relative entropy functional; all results stated in this section are well-known in the literature.

\subsection{Some fundamentals}\label{subsec:vonNeumann}

Let $\MH$ be a Hilbert space. A \emph{von Neumann algebra} is a strongly (equivalently, weakly) closed $\ast$-subalgebra $\MFM \subset \BO(\MH)$ of bounded linear operators on $\MH$ containing the identity $\id_\MH$. Since $\BO(\MH)$ is isomorphic to the dual space of $(\NO(\MH), \ndot_\mathrm{tr})$, one can also consider the weak-$\ast$ topology $\sigma (\BO(\MH), \NO(\MH))$ on $\MFM$, referred to as the \emph{$\sigma$-weak topology}, and show that every von Neumann algebra is $\sigma$-weakly closed.

Let $\MFM_\ast \subset \MFM^\ast$ denote the set of all $\sigma$-weakly continuous linear functionals $\varphi : \MFM \to \C$. A positive element $\varphi \in \MFM_\ast$, that is, a functional satisfying $\varphi(A^\ast A) \ge 0$ for all $A \in \MFM$, is called a \emph{normal functional}, and the set of these is denoted by the symbol $\NF{\MFM}$; a normal functional $\varphi$ with $\varphi(\id_\MH) = 1$ is called a \emph{normal state}, and their set is denoted by $\NS{\MFM}$. Given a vector $\xi \in \MH$, the functional $\omega_\xi : \MFM \to \C$ defined by
\begin{equation*}
    \omega_\xi(A) \ce \braket{\xi, A \xi} \com A \in \MFM \, ,
\end{equation*}
is called vector functional induced by $\xi$, and $\xi$ is called its \emph{vector representative}; it holds that $\omega_\xi \in \NF{\MFM}$ and $\omega_\xi \in \NS{\MFM}$ iff $\norm{\xi} = 1$. The \emph{support} of $\varphi \in \NF{\MFM}$, denoted by $\ssupp(\varphi)$ or $\ssupp_\varphi$, is defined to be the smallest projection in $\MFM$ satisfying
\begin{equation*}
    \varphi (\ssupp_\varphi) = \varphi(\id_\MH) \, .
\end{equation*}
Let $\comm{\MFM}$ denote the \emph{commutant} of $\MFM$, that is, the set of all $B \in \BO(\MH)$ which commute with every $A \in \MFM$. Then the support $\ssupp_\xi \ce \ssupp(\omega_\xi)$ of the vector functional $\omega_\xi = \braket{\xi, \cdot\, \xi}$, $\xi \in \MH$, is determined as follows \cite[Sect. 5.22]{StratilaZsido2019}:
\begin{equation*}
    \ran (\ssupp_\xi) = \oln{\comm{\MFM} \xi} \, .
\end{equation*}

A closed, densely defined linear operator $T : \MH \supset \dom(T) \to \MH$ is said to be \emph{affiliated with} the von Neumann algebra $\MFM \subset \BO(\MH)$, and the set of these operators is denoted by $\AE{\MFM}$, if for all $A^\prime \in \comm{\MFM}$ there holds
\begin{equation*}
    A^\prime \dom(T) \subset \dom(T) \tand A^\prime T = T A^\prime \ \text{on $\dom(T)$} \, .
\end{equation*}
If $T$ is self-adjoint, this is to say that $T \in \AE{\MFM}$ if and only if $T$ strongly commutes with all elements of the commutant $\comm{\MFM}$. Moreover, if $T \in \AE{\MFM}$ is self-adjoint and $f : \R \to \C$ is a bounded Borel-measurable function, it holds that $f(T) \in \AE{\MFM}$ \cite[p. 226]{StratilaZsido2019}.

\subsection{Tomita-Takesaki modular theory}\label{subsec:modularTheory}

Let $\MFM$ be a von Neumann algebra, and suppose that there exists a faithful $\omega \in \NF{\MFM}$, i.e., $\omega(A^\ast A) = 0$ implies $A = 0$ for all $A \in \MFM$; in this case, $\MFM$ is called \emph{$\sigma$-finite}. Passing to the GNS representation of $\MFM$ with respect to $\omega$, we may assume that $\MFM$ is represented on a Hilbert space $\MH$, and that there exists a cyclic and separating vector $\Omega \in \MH$ (meaning that $\oln{\MFM \Omega} = \MH$ and $A \Omega = 0$ implies $A = 0$ for all $A \in \MFM$) such that $\omega = \omega_\Omega = \braket{\Omega, \cdot\, \Omega}$ \cite[p. 13]{Hiai2021}.

The following construction is due to Tomita and Takesaki \cite{Takesaki1970}, constituting the basis of modular theory in von Neumann algebras; see also, e.g., Refs.~\cite{Hiai2021, Stratila2020, StratilaZsido2019} for exhaustive treatments. Consider the densely defined \emph{Tomita operator} $S_\Omega : \MH \supset \MFM \Omega \to \MH$ given by
\begin{equation}\label{eq:TomitaOperator}
    S_{\Omega} (A \Omega) = A^\ast \Omega \com A \in \MFM \, .
\end{equation}
It holds that $S_\Omega$ is anti-linear and closable. Denoting the closure by the same symbol, we can write its polar decomposition as
\begin{equation}\label{eq:TomitaPolarDecomp}
    S_\Omega = J_\Omega \Delta_\Omega^{1/2} \, .
\end{equation}
$J_\Omega$ is an anti-linear unitary involution ($J_\Omega^2 = \id$ and $J_\Omega^\ast = J_\Omega$), called \emph{modular conjugation}, and $\Delta_\Omega$ is a positive self-adjoint linear operator, called \emph{modular operator}, satisfying \cite[p. 14]{Hiai2021}
\begin{equation}\label{eq:modularData}
    J_\Omega \Omega = \Omega \com \Delta_\Omega \Omega = \Omega \tand J_\Omega \Delta_\Omega J_\Omega = \Delta_\Omega^{-1} \, .
\end{equation}
The fundamental result in modular theory is Tomita's theorem, which actually holds for arbitrary (not necessarily $\sigma$-finite) von Neumann algebras \cite[Thm. 2.2]{Hiai2021}.

\begin{theorem}[Tomita]\label{thm:Tomita}
    Let $\MFM \subset \BO(\MH)$ be a $\sigma$-finite von Neumann algebra with cyclic and separating vector $\Omega \in \MH$. Then
    \begin{equation*}
        J_\Omega \MFM J_\Omega = \comm{\MFM} \quad \text{and} \quad \Delta_\Omega^{\ii t} \MFM \Delta_\Omega^{- \ii t} = \MFM \quad (t \in \R) \, .
    \end{equation*}
\end{theorem}

\subsection{Standard form representation}\label{subsec:standardForm}

For the discussion of the relative modular operator as well as the development of perturbation theory, the standard form representation of von Neumann algebras, developed independently by H.~Araki \cite{Araki1974}, A.~Connes \cite{Connes1974}, and U.~Haagerup \cite{Haagerup1975}, is an indispensable tool.

A von Neumann algebra in \emph{standard form} is a quadruple $(\MFM, \MH, J, \NPC)$ consisting of a von Neumann algebra $\MFM$ represented faithfully on $\MH$, an anti-linear unitary involution $J : \MH \to \MH$, and a closed self-dual cone $\NPC \subset \MH$ (called \emph{natural positive cone}) satisfying
\begin{enumerate}
    \item \label{enu:standardFormTomita} $J \MFM J = \comm{\MFM}$;
    \item \label{enu:standardFormJCenter} $J A J = A^\ast$ for all $A \in \MFM \cap \comm{\MFM}$;
    \item \label{enu:standardFormJIdentityNPC} $J \xi = \xi$ for all $\xi \in \NPC$;
    \item \label{enu:standardFormJConjugationNPC} $AJAJ \NPC \subset \NPC$ for all $A \in \MFM$.
\end{enumerate}

For a $\sigma$-finite $\MFM$ with cyclic and separating vector $\Omega \in \MH$, one can take $J$ to be the modular conjugation $J_\Omega$ from \cref{subsec:modularTheory} and define $\NPC \ce \oln{\set{A J A \Omega : A \in \MFM}}$; thus, every $\sigma$-finite von Neumann algebra possesses a standard form representation \cite[Thm. 3.2]{Hiai2021}. More generally, Haagerup \cite{Haagerup1975} showed that every von Neumann algebra possesses a standard form and that the quadruple $(\MFM, \MH, J, \NPC)$ is unique up to unitary equivalence \cite[Sect. 10.26]{StratilaZsido2019}.

The following simple fact \cite[Lem. 3.9]{Hiai2021} is used often below.

\begin{lemma}\label{lem:cyclicSeparating}
    Let $(\MFM, \MH, J, \NPC)$ be a von Neumann algebra in standard form and $\Omega \in \NPC$. Then $\Omega$ is cyclic for $\MFM$ if and only if $\Omega$ is separating for $\MFM$.
\end{lemma}

The next very important theorem \cite[Thm. 3.12]{Hiai2021}, \cite[Thm. 10.25]{StratilaZsido2019}, relating the cone $\NPC$ to the set $\NF{\MFM}$, is due to Haagerup.

\begin{theorem}\label{thm:vectorRepresentation}
    Let $(\MFM, \MH, J, \NPC)$ be a standard form. Then the mapping $\NPC \owns \xi \mapsto \omega_\xi \in \NF{\MFM}$ is a homeomorphism with respect to the norm topologies. In particular, every normal functional $\varphi \in \NF{\MFM}$ is represented by a vector $\xi_\varphi \in \NPC$ from the natural positive cone.
\end{theorem}

Another important property of the standard form $(\MFM, \MH, J, \NPC)$ is that there exists a unitary representation of the group $\Aut(\MFM)$ of $\ast$-automorphisms of $\MFM$ on the Hilbert space $\MH$; see, e.g., Ref.~\cite[Cor. 2.5.32]{BR1}. Denote by $\UO(\MH)$ the group of unitary operators on $\MH$ equipped with the strong operator topology, and equip $\Aut(\MFM)$ with the topology of pointwise $\sigma$-weak convergence.

\begin{theorem}\label{thm:unitaryReprAut}
    Let $(\MFM, \MH, J, \NPC)$ be a von Neumann algebra in standard form. Then there exists a unique continuous unitary representation $\Aut(\MFM) \owns \alpha \mapsto U_\alpha \in \UO(\MH)$ with the following properties:
    \begin{enumerate}
        \item $\alpha(A) = U_\alpha A U_\alpha^\ast$ for all $A \in \MFM$;
        \item $U_\alpha \NPC \subset \NPC$;
        \item $U_\alpha J = J U_\alpha$.
    \end{enumerate}
\end{theorem}

The following key example for a standard form, taken from Ref.~\cite[Exa. 3.6 (2)]{Hiai2021}, will be resorted to several times throughout the paper.

\begin{example}\label{exa:standardForm}
    Let $\MFM = \BO(\MH)$. This algebra acts faithfully on the Hilbert space $\HS(\MH)$ of Hilbert-Schmidt operators [whose inner product is $\braket{X, Y}_\mathrm{HS} = \tr(X^\ast Y)$] by left multiplication, $\pi(A) X \ce A X$, and one can show that the standard form representation of $\MFM$ is given by
    \begin{equation*}
        \bigl(\BO(\MH), \HS(\MH), J A = A^\ast, \HS(\MH)_+\bigr) \, ,
    \end{equation*}
    where $J : \HS(\MH) \to \HS(\MH)$, $X \mapsto X^\ast$, is given by the operation of taking adjoints, and $\HS(\MH)_+$ is the cone of positive Hilbert-Schmidt operators.
\end{example}

\subsection{Relative modular operator}

The construction of the modular operator with respect to a von Neumann algebra $\MFM \subset \BO(\MH)$ and a cyclic and separating vector $\Omega \in \MH$ can be generalized to the situation in which two (not necessarily faithful) positive normal functionals are given; this is due to H.~Araki \cite{Araki1977}.

Let $(\MFM, \MH, J, \NPC)$ be a standard form and $\varphi, \psi \in \NF{\MFM}$. According to \cref{thm:vectorRepresentation}, there exist $\Phi, \Psi \in \NPC$ such that $\varphi = \omega_\Phi$ and $\psi = \omega_\Psi$. Define an operator $S_{\Psi, \Phi}^0 : \MH \supset \dom(S_{\Psi, \Phi}^0) \to \MH$ by
\begin{equation*}
    S_{\Psi, \Phi}^0(A \Phi + \eta) \ce \ssupp_\Phi A^\ast \Psi
\end{equation*}
for $A \in \MFM$ and $\eta \in (\id_\MH - \ssupp_\Phi^\prime) \MH$, where $\ssupp_\Phi^\prime \in \comm{\MFM}$ is the support of $\omega_\Phi^\prime (A^\prime) \ce \braket{\Phi, A^\prime \Phi}$, $A^\prime \in \comm{\MFM}$, with $\ran (\ssupp_\Phi^\prime) = \oln{\MFM \Phi}$. It holds that $J \ssupp_\Phi J = \ssupp_\Phi^\prime$ \cite[p. 163]{Hiai2021}, and the domain of $S_{\Psi, \Phi}^0$ is given by
\begin{equation*}
    \dom(S_{\Psi, \Phi}^0) = \MFM \Phi + [\MFM \Phi]^\perp \, .
\end{equation*}
$S_{\Psi, \Phi}^0$ is densely defined, anti-linear, and closable \cite[Lem. 10.2]{Hiai2021}. Let $S_{\Psi, \Phi}$ denote its closure. (Note that if $\psi = \varphi$ [and hence $\Psi = \Phi$] is faithful, then $S_{\Phi, \Phi}$ agrees with the operator $S_\Phi$ defined in \cref{subsec:modularTheory}.) The \emph{relative modular operator} $\Delta_{\Psi, \Phi}$ is defined to be
\begin{equation*}
    \Delta_{\Psi, \Phi} \ce S_{\Psi, \Phi}^\ast S_{\Psi, \Phi} \, .
\end{equation*}

The following lemma shows how the relative modular operator behaves under scaling of the normal functionals \cite[p. 179]{Araki1977}; it will be needed later on in \cref{sec:extension}.

\begin{lemma}\label{lem:relModOpScaling}
    Let $\lambda, \mu \in (0, + \infty)$. Then
    \begin{equation*}
        \Delta_{\sqrt{\mu} \, \Psi, \sqrt{\lambda} \, \Phi} = \frac{\mu}{\lambda} \, \Delta_{\Psi, \Phi} \, .
    \end{equation*}
\end{lemma}

We also heavily rely on the next result about convergence properties of the relative modular operator, taken from and proved in Ref.~\cite[Thm. 4.2]{DJP03}.

\begin{lemma}\label{lem:convRelModOp}
    Let $(\MFM, \MH, J, \NPC)$ be a standard form and $(\Phi_n)_{n \in \N} \subset \NPC$, $(\Psi_n)_{n \in \N} \subset \NPC$ be two sequences. Assume that the following properties are satisfied:
    \begin{enumerate}[label=\normalfont(\greek*)]
        \item $\Delta_{\Phi_n, \Psi_n} \to M$ in the strong resolvent sense;
        \item $\ssupp_{\Psi_n} \to \ssupp_\Psi$ in the strong operator topology;
        \item $\Phi_n \to \Phi$ and $\Psi_n \to \Psi$ weakly in $\MH$.
    \end{enumerate}
    Then it follows that $M = \Delta_{\Phi, \Psi}$.
\end{lemma}

\subsection{Araki-Uhlmann relative entropy}

Let $(\MFM, \MH, J, \NPC)$ be a von Neumann algebra in standard form. H.~Araki \cite{Araki1976, Araki1977} defined a relative entropy for two arbitrary normal functionals of $\MFM$, which generalizes H.~Umegaki's relative entropy, cf. \cref{eq:UmegakiRelEnt}; an equivalent definition of the relative entropy in a different setting was given independently by A.~Uhlmann \cite{Uhlmann1977}.

Let $\psi, \varphi \in \NF{\MFM}$ be two normal functionals on $\MFM$ and $\Psi, \Phi \in \NPC$ be their vector representatives (\cref{thm:vectorRepresentation}). The \emph{Araki-Uhlmann relative entropy} $\MS_\MFM(\psi, \varphi)$ is defined to be
\begin{equation}\label{eq:relativeEntropy}
    \MS_\MFM(\psi, \varphi) \ce
    \begin{cases}
        - \braket[\big]{\Psi, \log (\Delta_{\Phi, \Psi}) \Psi} & \text{if $\ssupp(\psi) \le \ssupp(\varphi)$} \, , \\
        + \infty & \text{otherwise} \, .
    \end{cases}
\end{equation}
The following important properties of the relative entropy are well-known; Araki proved assertions \ref{enu:relEntScaling}, \ref{enu:relEntLowerBound}, \ref{enu:relEntNonNeg} \cite[Lem. 3.2 \& Thm. 3.6]{Araki1977} as well as \ref{enu:relEntSubalg} for special cases \cite[Thm. 3.8]{Araki1977}, while the latter in its present form is a consequence of Uhlmann's celebrated monotonicity theorem \cite[Prop. 18]{Uhlmann1977}; see also Refs.~\cite[Prop. 5.1 \& Thm. 5.3]{OP2004} and \cite[Thm. 1 \& Cor. 2]{Petz1986a}.

\begin{theorem}\label{thm:relEntProperties}
    Let $\psi, \varphi \in \NF{\MFM}$, $\lambda, \mu > 0$, and $\MFM_0 \subset \MFM$ be a von Neumann subalgebra.
    \begin{enumerate}
        \item \label{enu:relEntScaling} $\displaystyle \MS_\MFM(\lambda \psi, \mu \varphi) = \lambda \, \MS_\MFM(\psi, \varphi) - \lambda \, \psi(\id_\MH) \log (\mu / \lambda)$.
        
        \item \label{enu:relEntLowerBound} $\displaystyle \MS_\MFM(\psi, \varphi) \ge - \psi(\id_\MH) \bigl(\log \bigl( \varphi(\ssupp_\psi)\bigr) - \log \bigl(\psi(\id_\MH)\bigr)\bigr)$.        

        \item \label{enu:relEntNonNeg} If $\psi(\id_\MH) = \varphi(\id_\MH) = 1$, then $\MS_\MFM(\psi, \varphi) \ge 0$. Moreover, $\MS_\MFM(\psi, \varphi) = 0$ iff $\psi = \varphi$.

        \item \label{enu:relEntSubalg} If $\psi(\id_\MH) = \varphi(\id_\MH) = 1$, then $\MS_\MFM (\psi, \varphi) \ge \MS_{\MFM_0} (\psi|_{\MFM_0}, \varphi|_{\MFM_0})$.
    \end{enumerate}
\end{theorem}

As a key example, which will be of relevance for us later, we shall compute the relative entropy for two normal functionals on $\BO(\MH)$.

\begin{example}\label{exa:relativeEntropy}
    Consider the von Neumann algebra $\MFM = \BO(\MH)$ and recall (cf. \cref{exa:standardForm}) that its standard form is given by $\bigl(\MFM, \HS(\MH), \,\cdot\,^\ast, \HS(\MH)_+\bigr)$. Using \cref{thm:vectorRepresentation}, for all $\psi, \varphi \in \NF{\MFM}$ one can find positive trace-class operators $\rho_\psi, \rho_\varphi \in \NO(\MH)_+$ such that $\Psi = \rho_{\psi}^{1/2}$ and $\Phi = \rho_\varphi^{1/2}$ are the vector representatives of $\psi, \varphi$ in $\HS(\MH)_+$.
        
    If $\rho_\varphi = \sum_{i \in \N} \lambda_i \, P_i$ and $\rho_\psi = \sum_{j \in \N} \mu_j \, Q_j$ are the spectral decompositions of $\rho_\psi$ and $\rho_\varphi$, then the relative modular operator $\Delta_{\Phi, \Psi}$ on $\HS(\MH)$ is given by \cite[Exa. 10.4 (2)]{Hiai2021}
    \begin{equation*}
        \Delta_{\Phi, \Psi} = \sum_{i,j=1}^{\infty} \lambda_i \inv{\mu_j} L_{P_i} R_{Q_j} \, ,
    \end{equation*}
    where $L_T$ and $R_T$ denote the left and right multiplication operators with the operator $T$, respectively. From functional calculus, one obtains
    \begin{align*}
        \log(\Delta_{\Phi, \Psi}) &= \sum_{i,j=1}^{\infty} \log(\lambda_i \inv{\mu_j}) L_{P_i} R_{Q_j} = \sum_{i,j=1}^{\infty} \log(\lambda_i) L_{P_i} R_{Q_j} - \sum_{i,j=1}^{\infty} \log(\mu_j) L_{P_i} R_{Q_j} \\
        &= \sum_{i=1}^{\infty} \log(\lambda_i) L_{P_i} - \sum_{j=1}^{\infty} \log(\mu_j) R_{Q_j} = \log(\rho_\varphi) - \log(\rho_\psi) \, .
    \end{align*}
    With this result, the following expression for the relative entropy $\MS_\MFM(\psi, \varphi)$ in case that $\supp(\rho_\psi) \subset \supp(\rho_\varphi)$ is found:
    \begin{align*}
        \MS_{\BO(\MH)}(\psi, \varphi) &= - \braket[\big]{\Psi, \log(\Delta_{\Phi, \Psi}) \Psi} = - \tr \Bigl(\rho_\psi^{1/2} \bigl(\log(\rho_\varphi) - \log(\rho_\psi)\bigr) \rho_\psi^{1/2}\Bigr) \\
        &= \tr\bigl(\rho_\psi \log \rho_\psi - \rho_\psi \log \rho_\varphi\bigr) \\[3pt]
        &= S(\rho_\psi, \rho_\varphi) \, .
    \end{align*}
    That is, the Araki-Uhlmann relative entropy of two positive normal functionals on $\MFM = \BO(\MH)$ reduces to the Umegaki relative entropy from \cref{eq:UmegakiRelEnt}.
\end{example}

\section{Perturbation theory in von Neumann algebras}\label{sec:perturbationTheory}

In the quantum-mechanical two-sided Bogoliubov inequality discussed in \cref{sec:BogoliubovQM}, key roles were played by the perturbation $U$ of the \enquote{free} Hamiltonian $H_0$ as well as the Gibbs states $\rho_0$ and $\rho$ corresponding to $H_0$ and $H = H_0 + U$. In arbitrary von Neumann algebras, Hamiltonians and, \emph{a fortiori}, Gibbs states do not exist in general. Therefore, we first have to introduce a suitable framework with which the two-sided Bogoliubov can be extended to this setting. The fundamental notion turns out to be that of a $W^\ast$-dynamics on a von Neumann algebra, extending the well-known concept of time evolution in quantum mechanics (which did not play a role in \cref{sec:BogoliubovQM} due to the availability of Hamiltonians). This concept allows for the definition of Liouvillians, replacing the Hamiltonian from quantum mechanics, KMS states, serving as a generalization of Gibbs states, and finally perturbations of Liouvillians and KMS states, compensating for the decomposition $H = H_0 + U$.

In the following, we will first define these concepts in \cref{subsec:dynamics}, then recall the foundations of bounded perturbation theory in \cref{subsec:boundedPerturbations}, after that review some more recent results pertaining to unbounded perturbations in \cref{subsec:unboundedPerturbations}, and finally prove some more auxiliary lemmas in \cref{subsec:perturbationRelModOp} that will be crucial for us in the next section.

\subsection{Dynamics, Liouvillians, and KMS states}\label{subsec:dynamics}

A \emph{$W^\ast$-dynamics} on a von Neumann algebra $\MFM \subset \BO(\MH)$ is a pointwise $\sigma$-weakly continuous one-parameter group $\tau = (\tau_t)_{t \in \R}$ of $\ast$-automorphisms $\tau_t \in \Aut(\MFM)$, that is, $\tau_0 = \id_\MFM$, $\tau_{s+t} = \tau_s \circ \tau_t$ for all $s, t \in \R$, and $t \mapsto \tau_t(A)$ is $\sigma$-weakly continuous for all $A \in \MFM$. The pair $(\MFM, \tau)$ is called \emph{$W^\ast$-dynamical system}.

Representing the von Neumann algebra in standard form $(\MFM, \MH, J, \NPC)$, \cref{thm:unitaryReprAut} yields, for every $t \in \R$, a unique unitary operator $U_t \in \UO(\MH)$ such that $\tau_t(A) = U_t A U_t^\ast$ for all $A \in \MFM$ and $U_t \NPC \subset \NPC$. Since $(U_t)_{t \in \R}$ is strongly continuous \cite[p. 149]{Pillet2006}, Stone's theorem (cf. \cref{thm:Stone}) implies that there exists a unique self-adjoint operator $L : \MH \supset \dom(L) \to \MH$, which is called the \emph{standard Liouvillian} of $\tau$, such that for all $t \in \R$,
\begin{equation*}
    U_t = \ee^{\ii t L} \, .
\end{equation*}

We note the following algebraic property of the standard Liouvillian; even though it is well-known (e.g., Ref.~\cite[Eq. (29)]{Pillet2006}), we provide a proof for the sake of completeness.

\begin{lemma}\label{lem:JLLJ}
    In the above setting, we have
    \begin{equation}\label{eq:JLLJ}
        J \dom(L) = \dom(L) \quad \text{and} \quad J L + L J = 0 \ \text{on $\dom(L)$} \, .
    \end{equation}
\end{lemma}

\begin{proof}
    By construction of the Liouvillian $L$ as the generator of the strongly continuous unitary group $(U_t)_{t \in \R}$ implementing the dynamics $\tau$, there holds (cf. \cref{pro:generatorUnitaryGroup})
    \begin{gather*}
        \dom(L) = \Set{\xi \in \MH \ : \ \lim_{h \to 0} \frac{1}{h} (U_h - \id) \xi \ \text{exists}} \, , \\
        L \xi = - \ii \lim_{h \to 0} \frac{1}{h} (U_h - \id) \xi \com \xi \in \dom(L) \, .
    \end{gather*}
    Let $\eta \in \dom(L)$ be arbitrary. From the relation $U_t J = J U_t$ [cf. \cref{thm:unitaryReprAut} (3)], it follows that the limit
    \begin{equation*}
        \lim_{h \to 0} \frac{1}{h} (U_h - \id) J \eta = J \lim_{h \to 0} \frac{1}{h} (U_h - \id) \eta
    \end{equation*}
    exists, hence $J \eta \in \dom(L)$, which shows that $J \dom(L) \subset \dom(L)$. Similarly, every $\eta \in \dom(L)$ can be written as $\eta = J (J \eta)$, and since $J \eta \in \dom(L)$ according to the previous observation, it also follows that $\dom(L) \subset J \dom(L)$; therefore, we have established the first equality in \eqref{eq:JLLJ}. Furthermore,
    \begin{align*}
        J L \eta = J \biggl(- \ii \lim_{h \to 0} \frac{1}{h} (U_h - \id) \eta\biggr) = \ii \lim_{h \to 0} \frac{1}{h} (U_h - \id) J \eta = - L J \eta
    \end{align*}
    for all $\eta \in \dom(L)$ because $J$ is anti-linear, hence the second equality in \eqref{eq:JLLJ} follows.
\end{proof}

The standard Liouvillian can be characterized as follows \cite[Thm. 2.9]{DJP03}, \cite[Prop. 4.46]{Pillet2006}.

\begin{proposition}\label{pro:standardLiouvillian}
    Let $((\MFM, \MH, J, \NPC), \tau)$ be a $W^\ast$-dynamical in standard form. The standard Liouvillian is the unique self-adjoint operator $L$ on $\MH$ satisfying for all $t \in \R$ and $A \in \MFM$:
    \begin{enumerate}
        \item $\ee^{\ii t L} \NPC \subset \NPC$;
        \item $\tau_t(A) = \ee^{\ii t L} A \, \ee^{- \ii t L}$.
    \end{enumerate}
\end{proposition}

Let $(\MFM, \tau)$ be a $W^\ast$-dynamical system and $\beta > 0$. A normal state $\omega \in \NS{\MFM}$ is called a \emph{$(\tau, \beta)$-KMS state} (introduced by Haag, Hugenholtz, and Winnink \cite{Haag1967} and named after Kubo, Martin, and Schwinger) if for all $A, B \in \MFM$ there exists a function $F_{\beta, A, B} : \C \to \C$ which is analytic in the strip $S_\beta = \set{z \in \C \, : \, 0 < \Im(z) < \beta}$ and continuous on its closure, and which satisfies the following \emph{KMS-boundary conditions} for all $t \in \R$:
\begin{equation*}
    F_{\beta, A, B}(t) = \omega\bigl(A \tau_t(B)\bigr) \tand F_{\beta, A, B}(t + \ii \beta) = \omega\bigl(\tau_t(B) A\bigr) \, .
\end{equation*}
We illustrate this rather abstract definition with two concrete examples that will become highly relevant for our investigations below.

\begin{examples}\label{exa:LiouvillianKMS}
    \leavevmode
    \begin{enumerate}[env]
        \item \label{enu:finiteQuantumSystem} Consider the von Neumann algebra $\MFM = \BO(\MH)$ and let $H : \MH \supset \dom(H) \to \MH$ be self-adjoint. Define a one-parameter group of $\ast$-automorphisms $\tau = (\tau_t)_{t \in \R}$ by setting
        \begin{equation*}
            \tau_t(A) \ce \ee^{\ii t H} A \, \ee^{- \ii t H} \, , \quad t \in \R \, , \ A \in \MFM \, .
        \end{equation*}
        The group $\tau$ is pointwise $\sigma$-weakly continuous \cite[Exa. 4.16]{Pillet2006}, hence, $(\MFM, \tau)$ is a $W^\ast$-dynamical system. According to \cref{exa:standardForm}, the standard form of $\MFM$ is $\bigl(\MFM, \HS(\MH), \,\cdot\,^\ast, \HS(\MH)_+\bigr)$, and one can show that the standard Liouvillian $L : \HS(\MH) \supset \dom(L) \to \HS(\MH)$ is determined by the following identity \cite[Exa. 4.51]{Pillet2006}
        \begin{equation*}
            \ee^{\ii t L} X = \ee^{\ii t H} X \, \ee^{- \ii t H} \com X \in \HS(\MH) \, .
        \end{equation*}

        Let $\beta > 0$ be arbitrary and assume that $\ee^{- \beta H} \in \NO(\MH)$. Define the density operator $\rho_\beta \ce \ee^{- \beta H} / \tr(\ee^{- \beta H})$ and the canonical Gibbs state $\varphi_\beta \ce \tr(\rho_\beta \,\cdot\,) \in \NS{\MFM}$. For any $A, B \in \MFM$ consider the function \cite[p. 21]{Hiai2021}
        \begin{equation*}
            F_{\beta, A, B}(z) \ce \frac{1}{\tr (\ee^{- \beta H})} \, \tr\bigl(\ee^{- \beta H} A \, \ee^{\ii z H} B \, \ee^{- \ii z H}\bigr) \, , \quad z \in \C \, .
        \end{equation*}
        It holds that $F_{\beta, A, B}$ is analytic in $S_\beta$ and continuous on $\oln{S_\beta}$, and one can compute
        \begin{equation*}
            F_{\beta, A, B}(t) = \varphi_\beta\bigl(A \tau_t(B)\bigr) \quad \text{and} \quad F_{\beta, A, B}(t + \ii \beta) = \varphi_\beta\bigl(\tau_t(B) A\bigr) \, .
        \end{equation*}
        Hence, the canonical Gibbs state $\varphi_\beta$ is a $(\tau, \beta)$-KMS state. Conversely, one can show that for \emph{every} $(\tau, \beta)$-KMS state $\omega = \tr(\rho_\omega \,\cdot\,) \in \NS{\MFM}$, $\rho_\omega \in \NO(\MH)$, of the $W^\ast$-dynamical system $(\MFM, \tau)$, it follows that $\rho_\omega = \ee^{- \beta H} / \tr(\ee^{- \beta H})$, i.e., $\omega$ is given by the canonical Gibbs state at inverse temperature $\beta$ \cite[Prop. 4.7]{Attal2006}, \cite[Exa. 5.5]{Pillet2006}.

        \item \label{enu:modularGroup} Let $\MFM \subset \BO(\MH)$ be a $\sigma$-finite von Neumann algebra with cyclic and separating vector $\Omega \in \MH$, standard form $(\MFM, \MH, J, \NPC)$, and modular operator $\Delta_\Omega$. Due to \cref{thm:Tomita}, one can define the \emph{modular group} $\sigma_t^\Omega : \MFM \to \MFM$ with respect to $\Omega$ by
        \begin{equation*}
            \sigma_t^\Omega (A) \ce \Delta_\Omega^{\ii t} \, A \, \Delta_\Omega^{- \ii t} \, , \quad t \in \R \, , \ A \in \MFM \, .
        \end{equation*}
        It holds that $\sigma^\Omega = (\sigma_t^\Omega)_{t \in \R}$ is a $W^\ast$-dynamics on $\MFM$ with standard Liouvillian $L_\Omega = \log(\Delta_\Omega)$ \cite[Thm. 2.4 (4)]{DJP03}. To see the latter, note that $\ee^{\ii t L_\Omega} \NPC = \Delta_\Omega^{\ii t} \NPC = \NPC$ \cite[Thm. 3.2 (iii)]{Hiai2021}. Moreover, for every $t \in \R$ and $A \in \MFM$,
        \begin{equation*}
            \ee^{\ii t L_\Omega} A \, \ee^{- \ii t L_\Omega} = \ee^{\ii t \log \Delta_\Omega} A \, \ee^{- \ii t \log \Delta_\Omega} = \Delta_\Omega^{\ii t} \, A \, \Delta_\Omega^{- \ii t} = \sigma_t^\Omega (A) \, .
        \end{equation*}
        Thus, according to \cref{pro:standardLiouvillian}, $L_\Omega = \log(\Delta_\Omega)$ must be the unique standard Liouvillian of the $W^\ast$-dynamical system $(\MFM, \sigma^\Omega)$.

        Regarding KMS states of the $W^\ast$-dynamical system $(\MFM, \sigma^\Omega)$, \emph{Takesaki's fundamental theorem} \cite[Thm. 2.14]{Hiai2021} asserts that $\omega = \omega_\Omega \in \NF{\MFM}$ is a $(\sigma^\Omega, -1)$-KMS state (with respect to the strip $\set{z \in \C \, : \, \beta < \Im(z) < 0}$), and that the modular group $\sigma^\Omega$ is uniquely determined as the $W^\ast$-dynamics on $\MFM$ for which $\omega$ satisfies the KMS condition at $\beta = - 1$.
    \end{enumerate}
\end{examples}

The following result, taken from Ref.~\cite[Thm. 2.13]{DJP03}, gives a characterization of the KMS condition expressed on the Hilbert-space level using the Liouvillian of the dynamics.

\begin{proposition}\label{pro:LiouvillianKMS}
    Let $(\MFM, \MH, J, \NPC)$ be a von Neumann algebra in standard form, $\Omega \in \NPC$ be a unit vector, and $\tau$ be a $W^\ast$-dynamics on $\MFM$ with standard Liouvillian $L$. Then $\omega_\Omega$ is a $(\tau, \beta)$-KMS state if and only if $\MFM \Omega \subset \dom(\ee^{- \beta L / 2})$ and
    \begin{equation*}
        \forall \, A \in \MFM \ : \ \ee^{- \beta L / 2} A \Omega = J A^\ast \Omega \, .
    \end{equation*}
    In this case, if $\Omega \in \NPC$ is also cyclic for $\MFM$, then it holds that
    \begin{equation*}
        \Delta_\Omega = \ee^{- \beta L} \, .
    \end{equation*}
\end{proposition}

\subsection{Bounded perturbation theory}\label{subsec:boundedPerturbations}

We will now briefly review some of the essential constructions of bounded perturbation theory in von Neumann algebras, as developed by H.~Araki \cite{Araki1973a, Araki1973b, Araki1973c}. This will set the stage to generalize some of these results to unbounded perturbations in the following \cref{subsec:unboundedPerturbations}.

Let $(\MFM, \MH, J, \NPC)$ be a von Neumann algebra in standard form and $\Phi \in \NPC$ be a separating vector for the algebra $\MFM$; then, by \cref{lem:cyclicSeparating}, $\Phi$ is also cyclic. For a self-adjoint element $Q \in \MFM$, one defines \cite[Eq. (1.1)]{Araki1973b}, \cite[Eq. (4.18)]{Araki1973c}
\begin{equation*}
    \Phi_Q \ce \sum_{n=0}^{\infty} \int_{0}^{1/2} \diff t_1 \int_{0}^{t_1} \diff t_2 \dotsb \int_{0}^{t_{n-1}} \diff t_n \, (\Delta_\Phi^{t_n} Q) (\Delta_\Phi^{t_{n-1}-t_n} Q) \dotsm (\Delta_\Phi^{t_1-t_2} Q) \, \Phi \, .
\end{equation*}
Araki showed that this series converges absolutely, and that the vector $\Phi_Q$ is again an element of $\NPC$ \cite[Prop. 4.1]{Araki1973c}. Moreover, he proved that $\Phi_Q$ is cyclic and separating for $\MFM$ \cite[Cor. 4.4]{Araki1973c}, and that $\Phi \in \dom \bigl(\ee^{(\log \Delta_\Phi + Q) / 2}\bigr)$ together with the identity \cite[Corollary on p. 170]{Araki1973b}
\begin{equation}\label{eq:boundedPerturbedVector}
    \Phi_Q = \ee^{(\log (\Delta_\Phi) + Q) / 2} \, \Phi \, .
\end{equation}
He used this relation between the unperturbed and perturbed vector to generalize the well-known Golden-Thompson trace inequality to arbitrary von Neumann algebras \cite[Thm. 2]{Araki1973b}:

\begin{proposition}[Golden-Thompson inequality]\label{pro:GoldenThompson}
    Let $(\MFM, \MH, J, \NPC)$ be a von Neumann algebra in standard form, let $\Phi \in \NPC$ be separating for $\MFM$, and let $Q \in \MFM$ be self-adjoint. Then
    \begin{equation}\label{eq:GoldenThompson}
        \norm[\big]{\ee^{(\log \Delta_\Phi + Q) / 2} \Phi} \le \norm[\big]{\ee^{Q / 2} \Phi} \, .
    \end{equation}
\end{proposition}

Later, in his study of relative entropy, Araki proved the following formulas for the relative modular operator involving the perturbed vector $\Phi_Q$ \cite[Eqs.~(4.28) and (4.29), p.~189]{Araki1977}.

\begin{proposition}\label{pro:boundedPerturbationRelModOp}
    Let $(\MFM, \MH, J, \NPC)$ be a von Neumann algebra in standard form, let $\Phi, \Psi \in \NPC$ be cyclic and separating, and let $Q \in \MFM$ be self-adjoint. It holds that
    \begin{equation}\label{eq:boundedPerturbationRelModOp}
        \log (\Delta_{\Phi_Q, \Psi}) = \log (\Delta_{\Phi, \Psi}) + Q \tand \log (\Delta_{\Psi, \Phi_Q}) = \log (\Delta_{\Psi, \Phi}) - J Q J \, .
    \end{equation}
\end{proposition}

If $\tau$ is a $W^\ast$-dynamics on $\MFM$, then every self-adjoint $Q \in \MFM$ also gives rise to a perturbed $W^\ast$-dynamics $\tau^Q$ on $\MFM$ given for all $t \in \R$ and $A \in \MFM$ by \cite[Prop. 5.4.1]{BR2}
\begin{equation}\label{eq:boundedPerturbedDynamics}
    \tau_t^Q(A) = \sum_{n=0}^{\infty} \ii^n \int_{0}^{t} \diff t_1 \int_{0}^{t_1} \diff t_2 \, \dotsb \int_{0}^{t_{n-1}} \diff t_n \, \Bigl[\tau_{t_n}(Q), \bigl[\dotsb, \bigl[\tau_{t_2}(Q), [\tau_{t_1}(Q), \tau_t(A)]\bigr] \dotsb\bigr]\Bigr] \, ,
\end{equation}
where the integrals exist in the $\sigma$-weak operator topology and the series converges in the norm topology of $\MFM$. In case that $\tau$ is implemented by a self-adjoint operator $H$, i.e., $\tau_t(A) = \ee^{\ii t H} A \, \ee^{- \ii t H}$, then the perturbed dynamics is given by \cite[Cor. 5.4.2]{BR2}
\begin{equation*}
    \tau_t^Q(A) = \ee^{\ii t (H + Q)} A \, \ee^{- \ii t (H + Q)} \, .
\end{equation*}

\subsection{Unbounded perturbation theory}\label{subsec:unboundedPerturbations}

More general perturbations than bounded elements $Q \in \MFM$ of the von Neumann algebra were considered by M.~J.~Donald \cite{Donald1987, Donald1990}, who developed a framework for semi-bounded perturbations modeled as extended-valued operators affiliated with the algebra; see Ref.~\cite[Ch. 12]{OP2004} for a review. Later, J.~Derezi\'{n}ski, V.~Jak\v{s}i\'{c}, and C.-A.~Pillet \cite{DJP03} developed a perturbation theory for general unbounded perturbations of KMS states; their method essentially relies on three natural assumptions on the self-adjoint perturbation $V \in \AE{\MFM}$, the use of the Golden-Thompson inequality \eqref{eq:GoldenThompson}, and weak convergence arguments (in contrast to strong convergence arguments employed by Araki) to establish the existence of a perturbed vector $\Phi_V$ analogous to \eqref{eq:boundedPerturbedVector}.

In the following, we will review some of the results of Derezi\'{n}ski, Jak\v{s}i\'{c}, and Pillet as well as extend them in a way suitable for our needs in \cref{sec:extension}. As indicated more precisely at the appropriate places below, we will make essential use of the methods developed by these authors to prove the desired extensions. Let us first introduce the following general hypotheses used throughout this entire subsection.

\begin{assumption}\label{asm:generalSetting}
    In what follows, $\MFM$ shall be a von Neumann algebra represented in standard form $(\MFM, \MH, J, \NPC)$, $\tau$ a $W^\ast$-dynamics on $\MFM$ with standard Liouvillian $L : \MH \supset \dom(L) \to \MH$, that is, $\tau_t(A) = \ee^{\ii t L} A \, \ee^{- \ii t L}$, and $V \in \AE{\MFM}$ a self-adjoint operator affiliated with $\MFM$.
\end{assumption}

For unbounded perturbations, it turns out to be relatively straightforward to prove the existence of the perturbed dynamics, generalizing \cref{eq:boundedPerturbedDynamics}. For this, the following assumption of Ref.~\cite[Assumption 3.1]{DJP03} is needed, which we will adopt in the sequel as well:
\begin{enumerate}[label=\normalfont\bfseries(A\arabic*), ref=\normalfont(A\arabic*)]
    \item \label{enu:A1} $L + V$ is essentially self-adjoint on $\dom(L) \cap \dom(V)$.
\end{enumerate}
To simplify notation, denote the self-adjoint closure $\oln{L + V}$ also by $L + V$. The following result is proved in Ref.~\cite[Thm. 3.3]{DJP03}.

\begin{proposition}\label{pro:unboundedPerturbedDynamics}
    In the setting of Assumptions \ref{asm:generalSetting} and \ref{enu:A1}, define
    \begin{equation}\label{eq:unboundedPerturbedDynamics}
        \tau_t^V(A) \ce \ee^{\ii t (L + V)} A \, \ee^{- \ii t (L + V)}
    \end{equation}
    for all $t \in \R$ and $A \in \MFM$. Then $t \mapsto \tau_t^V$ is a $W^\ast$-dynamics on $\MFM$ which, in the case that $V \in \MFM$ is bounded, coincides with the expression \eqref{eq:boundedPerturbedDynamics} for bounded perturbations.
\end{proposition}

Before we continue with the next result of Ref.~\cite{DJP03}, we prove a simple consequence of Assumption \ref{enu:A1} regarding the operator $L - J V J$, which we will need in \cref{sec:extension}; the proof relies on \cref{lem:JLLJ}.

\begin{lemma}\label{lem:LJVJessSA}
    Suppose that Assumptions \ref{asm:generalSetting} and \ref{enu:A1} are satisfied. Then the operator $L - J V J$ is essentially self-adjoint on $\dom(L) \cap \dom(J V J)$.
\end{lemma}

\begin{proof}
    For brevity, introduce the notation $W \ce J V J$. Then
    \begin{equation*}
        \dom(W) = J \dom(V) \, .
    \end{equation*}
    To see this, first let $\eta \in \dom(W) = \set{\xi \in \MH : J \xi \in \dom(V)}$. Since $J^2 = \id_\MH$, it follows that $\eta = J (J \eta)$ and hence $\eta \in J \dom(V)$. Conversely, if $\eta \in J \dom(V)$, then $\eta = J \xi$ for some $\xi \in \dom(V)$. Therefore, $J \eta = J^2 \xi = \xi \in \dom(V)$, and so $\eta \in \dom(W)$. In \cref{lem:JLLJ}, we have shown that $\dom(L) = J \dom(L)$ and $J L + L J = 0$ on $\dom(L)$. From this one obtains
    \begin{equation*}
        L - W = - J L J - J V J = - J (L + V) J \, ,
    \end{equation*}
    where the domain of this operator is given by
    \begin{equation*}
        \dom(L - W) = \dom(L) \cap \dom(W) = J \bigl(\dom(L) \cap \dom(V)\bigr) \, .
    \end{equation*}
    If $\eta = J \xi \in \dom(L - W)$ for $\xi \in \dom(L + V)$, then the action of $L - W$ can be written as $(L - W) \eta = - J (L + V) J \eta = - J (L + V) \xi$. By Assumption \ref{enu:A1}, the operator $L + V$ is essentially self-adjoint on $\dom(L) \cap \dom(V)$. Since $J$ is a bounded, self-adjoint, bijective operator, $L - W$ is essentially self-adjoint on $\dom(L) \cap \dom(W)$.
\end{proof}

Coming back to Ref.~\cite{DJP03}, in their next step the authors determined the standard Liouvillian of the perturbed dynamics $\tau^V$ from \cref{eq:unboundedPerturbedDynamics}. Define
\begin{equation}\label{eq:perturbedLiouvillian}
    L_V \ce L + V - J V J
\end{equation}
on its canonical domain and suppose that the following assumption \cite[Assumption 3.2]{DJP03} is satisfied in addition to \ref{enu:A1}:
\begin{enumerate}[label=\normalfont\bfseries(A\arabic*), ref=\normalfont(A\arabic*)]
    \setcounter{enumi}{1}
    \item \label{enu:A2} $L_V$ is essentially self-adjoint on $\dom(L) \cap \dom(V) \cap \dom(J V J)$.
\end{enumerate}
As before, the self-adjoint closure $\oln{L_V}$ will be denoted simply by $L_V$. The following result about this operator is proved in Ref.~\cite[Thm. 3.5]{DJP03}.

\begin{proposition}\label{pro:perturbedLiouvillian}
    In the setting of \cref{asm:generalSetting} together with \ref{enu:A1} and \ref{enu:A2}, it holds that $L_V$ is the standard Liouvillian of the perturbed dynamics $\tau^V$ from \cref{eq:unboundedPerturbedDynamics}.
\end{proposition}

In order to obtain their main result---the construction of perturbed KMS states---Derezi\'{n}ski, Jak\v{s}i\'{c}, and Pillet introduced a third and final assumption regarding the perturbation $V \in \AE{\MFM}$ \cite[Assumption 5.1]{DJP03}. Let $\beta > 0$ and $\Omega \in \NPC$ be separating (and hence cyclic) for the algebra $\MFM$. Suppose that\footnote{See \cref{subsec:extensionPerturbation} for an intuitive consequence of this assumption.}
\begin{enumerate}[label=\normalfont\bfseries(A\arabic*), ref=\normalfont(A\arabic*)]
    \setcounter{enumi}{2}
    \item \label{enu:A3} $\norm[\big]{\ee^{- \beta V / 2} \Omega} < + \infty$.
\end{enumerate}
Finiteness of this quantity is central for the construction procedure that leads to a perturbed vector and various nice properties of this object \cite[Thm 5.5]{DJP03}. The central idea of the proof is to approximate $V$ by a suitable sequence $(V_n)_{n \in \N} \subset \MFM$ of bounded operators, and then to use the Golden-Thompson inequality together with a weak convergence argument to show that a relation similar to \cref{eq:boundedPerturbedVector} extends to the unbounded case.

Let us quickly recall the approximation procedure \cite[p. 479]{DJP03}, which we will also need for our extension below. Denote by $E_V$ the unique spectral measure of the self-adjoint operator $V \in \AE{\MFM}$. For all $n \in \N$, define the function $f_n \ce \1_{[-n, n]} \cdot \id_{\sigma(V)}$ and the operator
\begin{equation}\label{eq:approximatingSequence}
    V_n \ce f_n(V) = \int_{\sigma(V) \cap [-n, n]} \lambda \diff E_V(\lambda) \, .
\end{equation}
Since $\norm{f_n}_\infty = n$, it follows that $V_n : \MH \to \MH$ is a bounded operator, and since $V$ is affiliated with the von Neumann algebra $\MFM$, it holds that $V_n \in \MFM$ for all $n \in \N$; the sequence $(V_n)_{n \in \N}$ will be called \emph{approximating sequence} for $V \in \AE{\MFM}$. From the dominated convergence theorem, one obtains the following well-known result \cite[Prop. 6.12]{Moretti19}.

\begin{lemma}\label{lem:convergenceUnboundedApproximation}
    Let $\xi \in \dom(V)$ be arbitrary. Then $\lim_{n \to \infty} V_n \xi = V \xi$ in $\MH$.
\end{lemma}

Based on this fact, one can establish convergence results for operators related to the Liouvillian $L$; note that the first two assertions are stated in Ref.~\cite[Thm. 5.6]{DJP03}, whereas the third one is new and needed later for our extension.

\begin{lemma}\label{lem:convergenceLiouvillian}
    In the setting of \cref{asm:generalSetting}, assume that the self-adjoint operator $V \in \AE{\MFM}$ satisfies \ref{enu:A1} and \ref{enu:A2}, and let $(V_n)_{n \in \N} \subset \MFM$ be an approximating sequence for $V$ as defined in \cref{eq:approximatingSequence}.
    \begin{enumerate}
        \item \label{enu:convUnperturbedL} $L + V_n \to L + V$ in the strong resolvent sense.
        \item \label{enu:convPerturbedL} $L_{V_n} \to L_V$ in the strong resolvent sense.
        \item \label{enu:convSum} $- L_{V_n} + V_n \to - L_V + V$ in the strong resolvent sense.
    \end{enumerate}
\end{lemma}

\begin{proof}
    For \ref{enu:convUnperturbedL} and \ref{enu:convPerturbedL} see Ref.~\cite[Thm. 5.6]{DJP03}.

    \ref{enu:convSum} By definition of the operator $L_V$ [cf. \cref{eq:perturbedLiouvillian}], it holds that $- L_{V_n} + V_n = - L + J V_n J$ for all $n \in \N$ and $- L_V + V = - L + J V J$. Since $L$, $J$, and $V_n$ are self-adjoint operators and the latter two are also bounded, the operator $- L + J V_n J$ is self-adjoint as well. According to \cref{lem:LJVJessSA}, the operator $- L + J V J$ is essentially self-adjoint on $\MD \ce \dom(L) \cap \dom(J V J) = J \bigl(\dom(L) \cap \dom(V)\bigr)$. Furthermore, because $J$ is an isometry and satisfies $J^2 = \id_\MH$, \cref{lem:convergenceUnboundedApproximation} implies that for all $\eta = J \xi \in \MD$:
    \begin{equation*}
        (- L_{V_n} + V_n) \eta = - L \eta + J V_n \xi \xrightarrow{n \to + \infty} - L \eta + J V \xi = (- L_V + V) \eta \, .
    \end{equation*}
    Thus, all assumptions of \cref{pro:coreSOimpliesSR} are satisfied, which implies that $- L_{V_n} + V_n \to - L_V + V$ in the strong resolvent sense as $n \to + \infty$.
\end{proof}

In the general setting of \cref{asm:generalSetting}, let us introduce the following set of perturbations $V \in \AE{\MFM}$ that satisfy \ref{enu:A1} -- \ref{enu:A3} with respect to the Liouvillian $L$, the cyclic and separating vector $\Omega \in \NPC$, and $\beta > 0$:
\begin{equation*}
    \MFS_1(L, \Omega) \ce \left\{V \in \AE{\MFM} \ : \ \parbox{0.5\textwidth}{\centering $V$ self-adjoint and satisfies assumptions \ref{enu:A1}, \ref{enu:A2}, and \ref{enu:A3} with respect to $L$ and $\Omega$}\right\} \, .
\end{equation*}
Derezi\'{n}ski, Jak\v{s}i\'{c}, and Pillet proved the following crucial result \cite[Thm. 5.5 (1)]{DJP03}, which generalizes Araki's construction \eqref{eq:boundedPerturbedVector} of the perturbed vector:
\begin{itemize}
    \item \itshape If $\omega \in \NS{\MFM}$ is a faithful $(\tau, \beta)$-KMS state ($\beta > 0$) with vector representative $\Omega \in \NPC$ and if $V \in \MFS_1(L, \Omega)$, then $\Omega \in \dom \bigl(\ee^{- \beta (L + V) / 2}\bigr)$.
\end{itemize}
In \cref{sec:extension} we will need a version of this assertion for non-KMS states as well, which is the content of the following proposition; since its proof uses exactly the same strategy that was devised in Ref.~\cite{DJP03}, we defer it to Appendix \ref{app:perturbedVectors}. There, we also show that our formulation is actually equivalent to the result stated above, cf. \cref{pro:equivalencePerturbedVectors}.

\begin{proposition}\label{pro:unboundedPerturbedVector}
    Let $(\MFM, \MH, J, \NPC)$ be a von Neumann algebra in standard form, let $\Omega \in \NPC$ be normalized and separating for $\MFM$, let $\beta > 0$, and define $L \ce - \inv{\beta} \log(\Delta_\Omega)$. If $V \in \MFS_1(L, \Omega)$, then it follows that
    \begin{equation}\label{eq:weakConvUnboundedPerturbedVector}
        \Omega \in \dom \bigl(\ee^{- \beta (L + V) / 2}\bigr) \tand \ee^{- \beta (L + V) / 2} \Omega = \dlim{w}{n \to \infty} \ee^{- \beta (L + V_n) / 2} \Omega \, .
    \end{equation}
\end{proposition}

\begin{definition}\label{def:unboundedPerturbedState}
    In the setting of \cref{pro:unboundedPerturbedVector}, define the perturbed vector $\Omega_{- \beta V}$ and corresponding state $\omega_V$ by
    \begin{equation}\label{eq:unboundedPerturbedVector}
        \Omega_{- \beta V} \ce \ee^{- \beta (L + V) / 2} \, \Omega \tand \omega_V \ce \frac{1}{\norm{\Omega_{-\beta V}}^2} \, \braket{\Omega_{-\beta V}, \,\cdot\ \Omega_{-\beta V}} \, .
    \end{equation}
\end{definition}

Note that this notation is consistent with Araki's definition \eqref{eq:boundedPerturbedVector} of the perturbed vector $\Phi_Q$ for bounded perturbations $Q \in \MFM$, and that using this notation, the weak convergence statement in \cref{eq:weakConvUnboundedPerturbedVector} can be expressed as
\begin{equation*}
    \Omega_{- \beta V} = \dlim{w}{n \to \infty} \Omega_{- \beta V_n} \, .
\end{equation*}

The following properties of $\Omega_{- \beta V}$, which we will need for our proofs in the next sections, are contained in Ref.~\cite[Thm. 5.5 (2) \& (3)]{DJP03} for KMS states; by virtue of \cref{pro:equivalencePerturbedVectors}, they automatically apply to the non-KMS case as well.

\begin{lemma}\label{lem:perturbedVectorSeparating}
    In the setting of \cref{pro:unboundedPerturbedVector}, it holds that $\Omega_{- \beta V} \in \NPC$ and, moreover, $\Omega_{- \beta V}$ is cyclic and separating for $\MFM$.
\end{lemma}
    
In Ref.~\cite[Thm. 5.5]{DJP03}, further important properties of the perturbed vector $\Omega_{- \beta V}$ and perturbed state $\omega_V$ were proved, assuming that $\omega = \omega_\Omega$ is a $(\tau, \beta)$-KMS state. We will build on these findings in the next section when extending the two-sided Bogoliubov inequality.

\begin{theorem}\label{thm:unboundedPerturbationsKMS}
    In addition to \cref{asm:generalSetting}, suppose that $\omega$ is a faithful $(\tau, \beta)$-KMS state with vector representative $\Omega \in \NPC$, let $\beta > 0$, and let $V \in \MFS_1(L, \Omega)$.
    \begin{enumerate}
        \item \label{enu:unboundedPerturbationKMS} $\omega_V$ is a faithful $(\tau^V, \beta)$-KMS state.
        \item \label{enu:unboundedPerturbationModOp} $\log(\Delta_{\Omega_{- \beta V}}) = - \beta L_V$.
        \item \label{enu:unboundedPerturbationRelModOp1} $\log(\Delta_{\Omega_{- \beta V}, \Omega}) = \log(\Delta_\Omega) - \beta V$.
        \item \label{enu:unboundedPerturbationRelEnt1} $\MS_\MFM(\omega, \omega_V) = \beta \omega(V) + \log\bigl(\norm{\Omega_{- \beta V}}^2\bigr)$.
        \item \label{enu:unboundedPerturbationPB} $\ee^{- \beta \braket{\Omega, V \Omega} / 2} \le \norm{\Omega_{- \beta V}}$.
        \item \label{enu:unboundedPerturbationGT} $\norm{\Omega_{- \beta V}} \le \norm{\ee^{- \beta V / 2} \Omega}$.
    \end{enumerate}
\end{theorem}

Note that the inequalities stated in \ref{enu:unboundedPerturbationPB} and \ref{enu:unboundedPerturbationGT} of the preceding theorem are the Peierls-Bogoliubov and Golden-Thompson inequalities, respectively, for unbounded perturbations, which extend the corresponding inequalities proved by Araki \cite[Thm. 1 and 2]{Araki1973b}.

\subsection{Perturbation of relative modular operators}\label{subsec:perturbationRelModOp}

To generalize the two-sided Bogoliubov inequality and the variational principles, we require additional formulas for the relative modular operator involving a perturbed vector, similar to the one of \cref{thm:unboundedPerturbationsKMS} \ref{enu:unboundedPerturbationRelModOp1}.

\begin{lemma}\label{lem:unboundedPerturbationRelModOp2}
    Let $(\MFM, \MH, J, \NPC)$ be a von Neumann algebra in standard form, let $\Omega \in \NPC$ be normalized and separating for $\MFM$, and let $L \ce - \inv{\beta} \log(\Delta_\Omega)$, $\beta > 0$. If $V \in \MFS_1(L, \Omega)$, then
    \begin{equation}\label{eq:unboundedPerturbationRelModOp2}
        \log(\Delta_{\Omega, \Omega_{- \beta V}}) = \log (\Delta_\Omega) + \beta J V J \, .
    \end{equation}
\end{lemma}

\begin{proof}
    For every $n \in \N$, let the self-adjoint operator $V_n \in \MFM$ be defined as in \cref{eq:approximatingSequence}. According to \cref{pro:boundedPerturbationRelModOp}, it holds that
    \begin{equation*}
        \log \bigl(\Delta_{\Omega, \Omega_{- \beta V_n}}\bigr) = \log(\Delta_\Omega) + \beta J V_n J = - \beta (L - J V_n J) \, .
    \end{equation*}
    \cref{lem:convergenceLiouvillian} \ref{enu:convSum} shows that $L - J V_n J \to L - J V J$ in the strong resolvent sense as $n \to + \infty$. Therefore, we also have that
    \begin{equation*}
        \Delta_{\Omega, \Omega_{- \beta V_n}} \to \ee^{- \beta (L - J V J)}
    \end{equation*}
    in the strong resolvent sense as $n \to + \infty$. Since $\Omega_{- \beta V_n}$ and $\Omega_{- \beta V}$ are separating for $\MFM$ (cf. \cref{lem:perturbedVectorSeparating}), it holds that $\ssupp(\Omega_{- \beta V_n}) = \ssupp(\Omega_{- \beta V}) = \id_\MH$. Finally, we have $\Omega_{- \beta V_n} \to \Omega_{- \beta V}$ weakly in $\MH$ as $n \to + \infty$ according to \cref{pro:unboundedPerturbedVector}. Thus, all the assumptions of \cref{lem:convRelModOp} are satisfied and we conclude that
    \begin{equation*}
        \Delta_{\Omega, \Omega_{- \beta V}} = \ee^{- \beta (L - J V J)} \, . \tag*{\qedhere}
    \end{equation*}
\end{proof}

The previous result \eqref{eq:unboundedPerturbationRelModOp2} generalizes the formula of Ref.~\cite[Thm. 5.1 (8)]{DJP03} to unbounded perturbations, as we will demonstrate explicitly now.

\begin{corollary}\label{cor:unboundedPerturbationRelModOp3}
    In the setting of \cref{lem:unboundedPerturbationRelModOp2}, suppose that $\omega = \omega_\Omega \in \NS{\MFM}$ is a faithful $(\tau, \beta)$-KMS state. Then
    \begin{equation}\label{eq:unboundedPerturbationRelModOp3}
        \log \bigl(\Delta_{\Omega, \Omega_{- \beta V}}\bigr) = \log \bigl(\Delta_{\Omega_{- \beta V}}\bigr) + \beta V \, .
    \end{equation}
\end{corollary}

\begin{proof}
    Since $\omega$ is a faithful $(\tau, \beta)$-KMS state, \cref{pro:LiouvillianKMS} implies that the Liouvillian $L$ of $\tau$ is given by $- \beta L = \log (\Delta_\Omega)$. Furthermore, we have $\log (\Delta_{\Omega_{- \beta V}}) = - \beta L_V$ according to \cref{thm:unboundedPerturbationsKMS} \ref{enu:unboundedPerturbationModOp}. Thus, using the previously established \cref{eq:unboundedPerturbationRelModOp2}, it follows that
    \begin{align*}
        \log(\Delta_{\Omega, \Omega_{- \beta V}}) &= - \beta L + \beta J V J \\
        &= - \beta ( L + V - J V J) + \beta V \\
        &= - \beta L_V + \beta V \\
        &= \log (\Delta_{\Omega_{- \beta V}}) + \beta V \, . \tag*{\qedhere}
    \end{align*}
\end{proof}

With the next result, we extend the first formula of Araki's \cref{pro:boundedPerturbationRelModOp} to unbounded perturbations; the following can also be seen as a generalization of \cref{thm:unboundedPerturbationsKMS} \ref{enu:unboundedPerturbationRelModOp1}.

\begin{lemma}\label{lem:unboundedPerturbedRelModOp4}
    Let $(\MFM, \MH, J, \NPC)$ be a von Neumann algebra in standard form, let $\Phi, \Psi \in \NPC$ be separating, define $L \ce - \inv{\beta} \log(\Delta_\Phi)$, and let $V \in \MFS_1(L, \Phi)$. Assume that the operator $\log(\Delta_{\Phi, \Psi}) - \beta V$ is essentially self-adjoint on $\MD \ce \dom(\log \Delta_{\Phi, \Psi}) \cap \dom(V)$. Then
    \begin{equation*}
        \log(\Delta_{\Phi_{- \beta V}, \Psi}) = \log(\Delta_{\Phi, \Psi}) - \beta V \, .
    \end{equation*}
\end{lemma}

\begin{proof}
    Let $(V_n)_{n \in \N} \subset \MFM$ be an approximating sequence for $V$ as defined in \cref{eq:approximatingSequence}. Since $V_n \xi \to V \xi$ in $\MH$ for all $\xi \in \MD$ according to \cref{lem:convergenceUnboundedApproximation}, the assumption on $\log(\Delta_{\Phi, \Psi}) - \beta V$ together with \cref{pro:coreSOimpliesSR} implies that
    \begin{equation*}
        \log(\Delta_{\Phi, \Psi}) - \beta V_n \to \log(\Delta_{\Phi, \Psi}) - \beta V
    \end{equation*}
    in the strong resolvent sense as $n \to + \infty$. For bounded perturbations, the first of Araki's identities in \cref{eq:boundedPerturbationRelModOp} entails
    \begin{equation*}
        \log(\Delta_{\Phi_{- \beta V_n}, \Psi}) = \log(\Delta_{\Phi, \Psi}) - \beta V_n \, .
    \end{equation*}
    Since the right-hand side converges to $\log(\Delta_{\Phi, \Psi}) - \beta V$ in the strong resolvent sense, and since $\Phi_{- \beta V_n} \to \Phi_{- \beta V}$ weakly by \cref{pro:unboundedPerturbedVector} and $\Phi_{- \beta V} \in \NPC$ according to \cref{lem:perturbedVectorSeparating}, it follows again from \cref{lem:convRelModOp} that $\log(\Delta_{\Phi_{- \beta V}, \Psi}) = \log(\Delta_{\Phi, \Psi}) - \beta V$.
\end{proof}

\section{Extension of the two-sided Bogoliubov inequality}\label{sec:extension}

We now have almost everything at hand to extend the two-sided Bogoliubov inequality to general von Neumann algebras. First, however, in \cref{subsec:extensionPerturbation} we need to formulate a fourth and final assumption on the perturbation $V \in \AE{\MFM}$, allowing us to derive an identity for the relative entropy similar to \cref{thm:unboundedPerturbationsKMS} \ref{enu:unboundedPerturbationRelEnt1}. Building on that, in \cref{subsec:twoSidedBogoliubov} we will establish the two-sided Bogoliubov inequality and, moreover, show that it reduces to \cref{eq:BogoliubovQM} in case of a type I factor. In \cref{subsec:bosonicField}, we shall discuss the extended inequality for a bosonic scalar field and an inner perturbation, and finally, in \cref{subsec:variationalBounds}, we will prove analogs of the variational principles from \cref{thm:variationalBoundsQM} in the von Neumann algebra setting.

\subsection{Additional assumption}\label{subsec:extensionPerturbation}

As hinted at above, our first goal is to establish a formula for the quantity $\MS_\MFM (\omega_V, \omega)$, analogous to the one for $\MS_\MFM (\omega, \omega_V)$ from \cref{thm:unboundedPerturbationsKMS} \ref{enu:unboundedPerturbationRelEnt1}; this will require an additional assumption on $V$ introduced below. Before that, notice the following consequence of Assumption \ref{enu:A3}, which is mentioned in Ref.~\cite[p. 479]{DJP03}.

Let $V \in \AE{\MFM}$ be a self-adjoint operator affiliated with $\MFM$ and $E_V$ be its spectral measure. The operator $V$ can be decomposed into the sum $V = V_+ + V_-$, where
\begin{equation*}
    V_+ \ce \1_{(0, + \infty)}(V) V \tand V_- \ce \1_{(- \infty, 0]}(V) V
\end{equation*}
are defined via the functional calculus. Let $\beta > 0$ and consider a faithful $(\tau, \beta)$-KMS state $\omega \in \NS{\MFM}$ with vector representative $\Omega \in \NPC$. Suppose that $V$ satisfies \ref{enu:A3}. It holds that
\begin{equation*}
    \beta \omega(V_-) = \int_{\sigma(V) \cap (- \infty, 0]} \beta \lambda \diff \mu_\Omega^V(\lambda) \, ,
\end{equation*}
where $\mu_\Omega^V (A) = \braket{\Omega, E_V (A) \Omega}$ for all Borel sets $A \subset \R$, as before. One can estimate the integral by using that\footnote{From $\ee^x = \lim_{n \to \infty} (1 + x / n)^n$ and the Bernoulli inequality $(1 + x)^n \ge 1 + nx$, valid for all $x \ge -1$ and $n \in \N$, it follows that $\ee^{x} \ge 1 + x$ for all $x \in \R$. (The validity of this inequality for $x < - 1$ is clear.) Hence, substituting $x \to - x$, one obtains $\ee^{-x} \ge 1 - x$, and so $- \ee^{-x} \le - 1 + x < x$.\label{ftn:perturabtionTheory_exponentialInequality}} $x > - \ee^{- x}$ for all $x \in \R$ and $\int_{\sigma(T) \cap (0, + \infty)} \ee^{- \beta \lambda} \diff \mu_\Omega^V(\lambda) \ge 0$:
\begin{equation*}
    \beta \omega(V_-) > - \int_{\sigma(V) \cap (- \infty, 0]} \ee^{- \beta \lambda} \diff \mu_\Omega^V(\lambda) \ge - \int_{\sigma(T)} \ee^{- \beta \lambda} \diff \mu_\Omega^V(\lambda) = - \norm{\ee^{- \beta V / 2} \Omega}^2 > - \infty \, .
\end{equation*}
The condition $\beta \omega(V_-) > - \infty$ implies that $\beta \omega(V) > - \infty$. Therefore, the expression for $\MS_\MFM (\omega, \omega_V)$ from \cref{thm:unboundedPerturbationsKMS} \ref{enu:unboundedPerturbationRelEnt1} cannot be $- \infty$.

Inspired by this observation, let us pose the following additional assumption on the self-adjoint operator $V \in \MFS_1(L, \Omega)$:
\begin{enumerate}[label=\normalfont\bfseries(A\arabic*), ref=\normalfont(A\arabic*)]
    \setcounter{enumi}{3}
    \item \label{enu:A4} $\beta \omega_V(V_+) < + \infty$.
\end{enumerate}
This condition, in particular, implies that $\omega_V(V) < + \infty$. To compactify the notation again, introduce the set
\begin{equation*}
    \MFS_2(L, \Omega) \ce \left\{V \in \AE{\MFM} \ : \ \parbox{0.45\textwidth}{\centering $V$ self-adjoint and satisfies assumptions \ref{enu:A1} -- \ref{enu:A4} with respect to $L$ and $\Omega$}\right\} \subset \MFS_1(L, \Omega) \, .
\end{equation*}

The following result complements \cref{thm:unboundedPerturbationsKMS} \ref{enu:unboundedPerturbationRelEnt1}.

\begin{lemma}\label{lem:unboundedPerturbationRelEnt2}
    Let $(\MFM, \MH, J, \NPC)$ be a von Neumann algebra in standard form, let $\tau$ be a $W^\ast$-dynamics on $\MFM$ with Liouvillian $L$, and let $\omega$ be a faithful $(\tau, \beta)$-KMS state ($\beta > 0)$ with vector representative $\Omega \in \NPC$. If $V \in \MFS_2(L, \Omega)$, then
    \begin{equation*}
        \MS_\MFM(\omega_V, \omega) = - \beta \omega_V(V) - \log \bigl(\norm{\Omega_{- \beta V}}^2\bigr) \, .
    \end{equation*}
\end{lemma}

\begin{proof}   
    Define the operator $W \ce V + \inv{\beta} \log\bigl(\norm{\Omega_{- \beta V}}^2\bigr) \, \id_\MH$, which is still an element of the set $\MFS_2(L, \Omega)$. The $W$-perturbation \eqref{eq:unboundedPerturbedVector} of the vector $\Omega$ is given by
    \begin{equation*}
        \Omega_{-\beta W} = \ee^{- \beta (L + W) / 2} \, \Omega = \ee^{- \log \norm{\Omega_{- \beta V}} \id_\MH} \, \ee^{- \beta (L + V)} \Omega = \frac{\Omega_{- \beta V}}{\norm{\Omega_{- \beta V}}} \, ,
    \end{equation*}
    and hence the corresponding state $\omega_W$ [cf. again \cref{eq:unboundedPerturbedVector}] satisfies
    \begin{align*}
        \omega_W (A) &= \frac{1}{\norm{\Omega_{- \beta W}}^2} \, \braket{\Omega_{- \beta W}, A \Omega_{- \beta W}} \\[4pt]
        &= \braket{\Omega_{- \beta W}, A \Omega_{- \beta W}} \\
        &= \frac{1}{\norm{\Omega_{- \beta V}}^2} \, \braket{\Omega_{- \beta V}, A \Omega_{- \beta V}} = \omega_V (A)
    \end{align*}
    for all $A \in \MFM$, that is, we have $\omega_W = \omega_V$ on $\MFM$. Using now the identity \eqref{eq:unboundedPerturbationRelModOp3} from \cref{cor:unboundedPerturbationRelModOp3}, it follows that
    \begin{equation*}
        \log (\Delta_{\Omega, \Omega_{- \beta W}}) = \log(\Delta_{\Omega_{- \beta W}}) + \beta W = \log(\Delta_{\Omega_{- \beta W}}) + \beta V + \log\bigl(\norm{\Omega_{- \beta V}}^2\bigr) \id_\MH \, .
    \end{equation*}
    Since both $\omega$ and $\omega_V$ are faithful (the latter according to \cref{lem:perturbedVectorSeparating}), the relative entropy $\MS_\MFM (\omega_V, \omega)$ defined in \cref{eq:relativeEntropy} can be written as
    \begin{align*}
        \MS_\MFM (\omega_V, \omega) &= \MS_\MFM (\omega_W, \omega) = - \braket[\big]{\Omega_{- \beta W}, \log (\Delta_{\Omega, \Omega_{- \beta W}}) \Omega_{- \beta W}} \\[4pt]
        &= - \braket[\big]{\Omega_{- \beta W}, \underbrace{\log(\Delta_{\Omega_{- \beta W}}) \Omega_{- \beta W}}_{=0}} - \beta \underbrace{\braket[\big]{\Omega_{- \beta W}, V \Omega_{- \beta W}}}_{= \omega_V(V)} \\[-2pt]
            &\quad - \log \bigl(\norm{\Omega_{- \beta V}}^2\bigr) \braket[\big]{\Omega_{- \beta W}, \Omega_{- \beta W}} \\[4pt]
        &= - \beta \omega_V (V) - \log \bigl(\norm{\Omega_{- \beta V}}^2\bigr) \, .
    \end{align*}
    Note that $\log(\Delta_{\Omega_{- \beta W}}) \Omega_{- \beta W} = 0$ by functional calculus and the fact that $\Delta_{\Omega_{- \beta W}} \Omega_{- \beta W} = \Omega_{- \beta W}$, cf. \cref{eq:modularData}. According to Assumption \ref{enu:A4}, it holds that $\beta \omega_V(V)$ is a finite number or $- \infty$. Therefore, the right-hand side in the above equation cannot be equal to $- \infty$, which is necessary to guarantee non-negativity of the relative entropy.
\end{proof}

\subsection{Two-sided Bogoliubov inequality}\label{subsec:twoSidedBogoliubov}

The last thing that is missing in order to generalize \cref{eq:BogoliubovQM} to arbitrary von Neumann algebras is to define an analog of the relative free energy $\Delta F$ that was introduced in \cref{eq:interfaceEnergy}.

\begin{definition}\label{def:relativeFreeEnergy}
    Let $(\MFM, \MH, J, \NPC)$ be a von Neumann algebra in standard form, let $\tau$ be a $W^\ast$-dynamics on $\MFM$ with Liouvillian $L$, and let $\omega \in \NS{\MFM}$ be a faithful $(\tau, \beta)$-KMS state ($\beta > 0$) with vector representative $\Omega \in \NPC$. Given a perturbation $U \in \MFS_2(L, \Omega)$, define the \emph{relative free energy} between $\omega_U$ and $\omega$ to be
    \begin{equation}\label{eq:relativeFreeEnergy}
        \MF(\omega_U, \omega) \ce \omega_U(U) + \inv{\beta} \MS_\MFM(\omega_U, \omega) \, .
    \end{equation}
\end{definition}

\begin{remark}
    Consider the algebra $\MFM = \BO(\MH)$ with Hamiltonian $H_0$ at inverse temperature $\beta > 0$, let $\rho_0$ be the canonical Gibbs state \eqref{eq:freeGibbsState} for $H_0$, and let $\rho$ be the corresponding state \eqref{eq:canonicalGibbsState} for $H = H_0 + U$. In \cref{eq:realtiveFreeEnergyFQS} in the proof of \cref{thm:BogoliubovQM}, it was mentioned that
    \begin{equation*}
        \Delta F = \EV_\rho[U] + \inv{\beta} S(\rho, \rho_0) \, ,
    \end{equation*}
    where $S(\rho, \rho_0)$ is the Umegaki relative entropy. The right-hand side can be defined in general von Neumann algebras, namely, in the form of \cref{eq:relativeFreeEnergy} which is, therefore, a proper extension of the concept of relative free energy $\Delta F$. In fact, in \cref{pro:BogoliubovVNimpliesQM} below we will show that $\MF(\omega_U, \omega)$ indeed reduces to $\Delta F$ when $\MFM = \BO(\MH)$.
    
    A definition of a relative free energy in operator algebras similar to \cref{eq:relativeFreeEnergy}, justified by physical principles and the analogy to finite quantum systems, was suggested by M.~J.~Donald \cite[Eq.~(1.2)]{Donald1987}; he considered a relative free energy of an arbitrary normal state with respect to an equilibrium state, both defined with respect to the same Hamiltonian of the quantum system. A formula very similar to \cref{eq:relativeFreeEnergy} also appeared in another context in the work of R.~Longo \cite[Eq. (3.6)]{Longo1997}.
\end{remark}

We can now finally prove the two-sided Bogoliubov inequality, which is one of the main results of this paper.

\begin{theorem}[Two-sided Bogoliubov inequality]\label{thm:BogoliubovVN}
    Let $(\MFM, \MH, J, \NPC)$ be a von Neumann algebra in standard form, let $\tau$ be a $W^\ast$-dynamics on $\MFM$ with standard Liouvillian $L$, let $\beta > 0$, and let $\omega \in \NS{\MFM}$ be a faithful $(\tau, \beta)$-KMS state with vector representative $\Omega \in \NPC$. For all perturbations $U \in \MFS_2(L, \Omega)$ there holds
    \begin{equation}\label{eq:BogoliubovVN}
        \omega_U(U) \le \MF(\omega_U, \omega) \le \omega(U) \, .
    \end{equation}
\end{theorem}

\begin{proof}
    According to \cref{lem:unboundedPerturbationRelEnt2}, we have $\MS_\MFM(\omega_U, \omega) = - \beta \omega_U(U) - \log\bigl(\norm{\Omega_{- \beta U}}^2\bigr)$. Therefore, the relative free energy \eqref{eq:relativeFreeEnergy} can be written as
    \begin{equation}\label{eq:perturbedVectorRFE}
        \MF(\omega_U, \omega) = - \inv{\beta} \log\bigl(\norm{\Omega_{- \beta U}}^2\bigr) \, .
    \end{equation}
    Using that the Araki-Uhlmann relative entropy $\MS_\MFM(\omega_U, \omega)$ is non-negative, cf. \cref{thm:relEntProperties} \ref{enu:relEntNonNeg}, it follows that $\omega_U(U) \le \MF(\omega_U, \omega)$, which is the lower bound. Similarly, combining $\MS_\MFM(\omega, \omega_U) = \beta \omega(U) + \log (\norm{\Omega_{- \beta U}}^2)$ from \cref{thm:unboundedPerturbationsKMS} \ref{enu:unboundedPerturbationRelEnt1} with non-negativity of $\MS_\MFM(\omega, \omega_U)$, one obtains the upper bound $\MF(\omega_U, \omega) \le \omega(U)$.
\end{proof}

\begin{remark}     
    As in the quantum-mechanical case (cf. \cref{rem:BogoliubovQM}), the upper bound $- \inv{\beta} \log (\norm{\Omega_{- \beta U}}^2) \le \omega(U)$ is a version of the \emph{Peierls-Bogoliubov inequality}. In von Neumann algebras, it was established for bounded perturbations by Araki \cite[Thm. 1]{Araki1973b}, and for unbounded $U \in \MFS_1(L, \Omega)$ by Derezi\'{n}ski, Jak\v{s}i\'{c}, and Pillet \cite[Thm. 5.5 (8)]{DJP03}, see \cref{thm:unboundedPerturbationsKMS} \ref{enu:unboundedPerturbationPB}. In its typical form, this inequality reads
    \begin{equation*}
        \ee^{- \beta \omega(U)} \le \norm{\Omega_{- \beta U}}^2 \, .
    \end{equation*}
    The lower bound $\omega_U(U) \le - \inv{\beta} \log (\norm{\Omega_{- \beta U}}^2)$, on the other hand, does not seem to have been discussed in the literature so far; a different lower bound is given by the well-known \emph{Golden-Thompson inequality}, which for unbounded perturbations $U \in \MFS_1(L, \Omega)$ was also established in Ref.~\cite[Thm. 5.5 (9)]{DJP03}, see \cref{thm:unboundedPerturbationsKMS} \ref{enu:unboundedPerturbationGT}:
    \begin{equation*}
        \norm{\Omega_{- \beta U}} \le \norm{\ee^{- \beta U / 2} \Omega} \, .
    \end{equation*}
    Thus, it has to be stressed that the lower bound in the two-sided Bogoliubov inequality \eqref{eq:BogoliubovVN} is a new estimate whose usefulness for applications in physics has already been established \cite{Reible2022, Reible2025apx, Reible2023, Reible2025camcos}.
\end{remark}

Our next goal is to show that \cref{eq:BogoliubovVN} can indeed be considered the proper generalization of the quantum-mechanical two-sided Bogoliubov inequality \eqref{eq:BogoliubovQM} for unbounded perturbations on general von Neumann algebras $\MFM$: we will prove that if $\MFM = \BO(\MH)$, then \eqref{eq:BogoliubovVN} reduces to \eqref{eq:BogoliubovQM}. This is, in fact, not trivial due to the use of the standard form representation in the proof of \cref{thm:BogoliubovVN}. The key technical tool for the proof is the following result.

\begin{lemma}\label{lem:perturbedLiouvillianFQS}
    Consider the von Neumann algebra $\MFM = \BO(\MH)$, let $H_0$ be a self-adjoint linear operator acting in $\MH$, let $U$ be relatively $H_0$-bounded such that $H = H_0 + U$ is again self-adjoint, and let $L$ be the Liouvillian of the dynamics generated by $H_0$, cf. \cref{exa:LiouvillianKMS} \ref{enu:finiteQuantumSystem}. Then for all $t \in \R$ and $X \in \HS(\MH)$ there holds
    \begin{equation}\label{eq:perturbedLiouvillianFQS}
        \ee^{\ii t (L + U)} X = \ee^{\ii t H} X \, \ee^{- \ii t H_0} \, .
    \end{equation}
\end{lemma}

\begin{proof}
    Recall from \cref{exa:standardForm} that the standard form representation of $\MFM$ is given in terms of left multiplication on the Hilbert space $\MK \ce \HS(\MH)$ of Hilbert-Schmidt operators on $\MH$, and from \cref{exa:LiouvillianKMS} \ref{enu:finiteQuantumSystem} that the standard Liouvillian $L : \MK \supset \dom(L) \to \MK$ of the $W^\ast$-dynamics $\tau_t(A) \ce \ee^{\ii t H_0} A \, \ee^{- \ii t H_0}$, $t \in \R$, $A \in \MFM$, is determined by
    \begin{equation}\label{eq:LiouvillianFQS}
        \ee^{\ii t L} X = \ee^{\ii t H_0} X \, \ee^{- \ii t H_0} \com X \in \MK \, .
    \end{equation}

    Since both $t \mapsto \ee^{\ii t H}$ and $t \mapsto \ee^{- \ii t H_0}$ are strongly continuous one-parameter unitary groups on the Hilbert space $\MH$ (\cref{pro:generatorUnitaryGroup}), it is straightforward to see that the one-parameter family $t \mapsto W(t)$ of operators $W(t) : \MK \to \MK$ given by
    \begin{equation*}
        W(t) X \ce \ee^{\ii t H} X \, \ee^{- \ii t H_0} \com t \in \R \, , \ X \in \MK \, ,
    \end{equation*}
    is well-defined and also a strongly continuous one-parameter unitary group on $\MK$. Let $K : \MK \supset \dom(K) \to \MK$ denote the self-adjoint generator of the group $t \mapsto W(t)$, which exists by Stone's theorem (\cref{thm:Stone}) and is characterized as follows (\cref{pro:generatorUnitaryGroup}):
    \begin{gather*}
        \dom(K) = \Set{X \in \MK \ : \ \odbound{}{t}{t=0} W(t) X = \lim_{h \to 0} \frac{1}{h} \, \bigl(W(h) X - X\bigr) \in \MK \ \text{exists}} \, , \\
        K X = - \ii \, \odbound{}{t}{t=0} W(t) X \com X \in \dom(K) \, .
    \end{gather*}
    In the following, we will show that the generator $K$ of $(W(t))_{t \in \R}$ is given by $\oln{L + U}$, which then implies our claim. To do that, we will \textbf{(1)} identify a core $\MD$ of $K$ and \textbf{(2)} show that these operators agree on $\MD$ to conclude that $K = \oln{L + U}$.

    \textbf{(1)} For $\xi, \eta \in \MH$ let $P_{\xi, \eta} : \MH \to \MH$ denote the finite-rank operator $P_{\xi, \eta} \zeta \ce \braket{\eta, \zeta} \, \xi$, $\zeta \in \MH$, and define the following subspace of $\MK$:
    \begin{equation*}
        \MD \ce \lin \set[\big]{P_{\xi, \eta} \ : \ \xi, \eta \in \dom(H_0)} \, .
    \end{equation*}
    We will show that
    \begin{enumerate}
        \item $\MD \subset \MK$ is dense;
        \item $\MD \subset \dom(K)$;
        \item $W(t) \MD \subset \MD$ for all $t \in \R$.
    \end{enumerate}
    \cref{pro:coreGenerator} then implies that $\MD$ is a core for the operator $K$. 

    \tAd{} (a). It suffices to note that the finite-rank operators on $\MH$ are dense in $\MK$ (see, e.g., Ref.~\cite[Cor. 26.4]{BlanchardBruening2015}), and that since $\dom(H_0)$ is dense in $\MH$ (because $H_0$ is self-adjoint), every finite-rank operator on $\MH$ can be approximated by a finite-rank operator on $\dom(H_0)$, that is, by an element from $\MD$.
    
    \tAd{} (b) and (c). First, note that $\dom(H) = \dom(H_0)$ since, by assumption, $U$ is relatively $H_0$ bounded which implies $\dom(H_0) \subset \dom(U)$ and, hence, $\dom(H) = \dom(H_0) \cap \dom(U) = \dom(H_0)$. Now, let $P_{\xi, \eta} \in \MD$ be arbitrary.  For all $\zeta \in \MH$, it holds that
    \begin{equation*}
        \bigl(W(t) P_{\xi, \eta}\bigr) \zeta =  \braket{\eta, \ee^{- \ii t H_0} \zeta} \, \ee^{\ii t H} \xi = \braket{\ee^{\ii t H_0} \eta, \zeta} \, \ee^{\ii t H} \xi = P_{\ee^{\ii t H} \xi, \, \ee^{\ii t H_0} \eta} \, \zeta \, .
    \end{equation*}
    Since $\xi, \eta \in \dom(H_0) = \dom(H)$, we also have $\ee^{\ii t H} \xi, \, \ee^{\ii t H_0} \eta \in \dom(H_0)$, cf. \cref{pro:generatorUnitaryGroup}, which implies $W(t) P_{\xi, \eta} = P_{\ee^{\ii t H} \xi, \, \ee^{\ii t H_0} \eta} \in \MD$, proving property (c). Moreover, from
    \begin{align*}
        \bigl(W(t) P_{\xi, \eta} - P_{\xi, \eta}\bigr) \zeta &= \braket{\ee^{\ii t H_0} \eta, \zeta} \, \ee^{\ii t H} \xi - \braket{\eta, \zeta} \, \xi \\
        &= \braket{\ee^{\ii t H_0} \eta, \zeta} \, (\ee^{\ii t H} \xi - \xi) + \braket[\big]{(\ee^{\ii t H_0} - \id) \, \eta, \zeta} \, \xi
    \end{align*}
    it follows that
    \begin{align*}
        \lim_{t \to 0} \frac{1}{t} \, \bigl(W(t) P_{\xi, \eta} - P_{\xi, \eta}\bigr) \zeta &= \lim_{t \to 0} \left[\braket{\ee^{\ii t H_0} \eta, \zeta} \, \frac{\ee^{\ii t H} \xi - \xi}{t}\right] + \Braket{\lim_{t \to 0} \frac{1}{t} \, (\ee^{\ii t H_0} - \id) \, \eta, \zeta} \, \xi \\
        &= \braket{\eta, \zeta} \, \ii H \xi + \braket{\ii H_0 \eta, \zeta} \, \xi \\
        &= \ii \bigl(P_{H \xi, \eta} - P_{\xi, H_0 \eta}\bigr) \zeta \, ,
    \end{align*}
    where again $\xi, \eta \in \dom(H_0) = \dom(H)$ was used to evaluate the limits of the unitary groups on $\MH$. Since the vector $\zeta \in \MH$ was arbitrary, this shows that
    \begin{equation*}
        \lim_{t \to 0} \frac{1}{t} \, \bigl(W(t) P_{\xi, \eta} - P_{\xi, \eta}\bigr) = \ii \bigl(P_{H \xi, \eta} - P_{\xi, H_0 \eta}\bigr) \, .
    \end{equation*}
    As the right-hand side is a well-defined element of $\MK$, we can conclude that $P_{\xi, \eta} \in \dom(K)$, and hence (b) is proved.

    Having verified properties (a) -- (c), we can now invoke \cref{pro:coreGenerator} which implies that the space $\MD$ is a core for the generator $K$ of the unitary group $W(t)$ on $\MK$, that is,
    \begin{equation}\label{eq:closureRestrictedGenerator}
        K = \oln{K \restr \MD} \, ,
    \end{equation}
    where $K \restr \MD$ is the restriction of the self-adjoint operator $K$ to the subspace $\MD$. The previous argument also shows that for all $P_{\xi, \eta} \in \MD$, the restriction is given by
    \begin{equation}\label{eq:restrictedGenerator}
        K P_{\xi, \eta} = P_{H \xi, \eta} - P_{\xi, H_0 \eta} \, .
    \end{equation}

    \textbf{(2)} Coming back to the standard Liouvillian $L$ which generates the group $\ee^{\ii t L}$ given by \cref{eq:LiouvillianFQS}, it follows by the same argument as before that $\MD \subset \dom(L)$ and that for $P_{\xi, \eta} \in \MD$:
    \begin{equation*}
        L P_{\xi, \eta} = - \ii \, \odbound{}{t}{t=0} \Bigl(\ee^{\ii t H_0} P_{\xi, \eta} \, \ee^{- \ii t H_0}\Bigr) = P_{H_0 \xi, \eta} - P_{\xi, H_0 \eta} \, .
    \end{equation*}
    Since $\dom(H_0) \subset \dom(U)$ by assumption, we also have $U P_{\xi ,\eta} = P_{U \xi, \eta}$. Therefore, it follows from \cref{eq:restrictedGenerator} that
    \begin{align*}
        K P_{\xi, \eta} &= P_{H \xi, \eta} - P_{\xi, H_0 \eta} \\
        &= P_{(H_0 + U) \xi, \eta} - P_{\xi, H_0 \eta} \\
        &= P_{H_0 \xi, \eta} + P_{U \xi, \eta} - P_{\xi, H_0 \eta} \\
        &= (L + U) P_{\xi, \eta} \, .
    \end{align*}
    This proves $K \restr \MD = (L + U) \restr \MD$. Note that in the third line, we used that $P_{(H_0 + U) \xi, \eta} \zeta = \braket{\eta, \zeta} \, (H_0 + U) \xi = \braket{\eta, \zeta} \, H_0 \xi + \braket{\eta, \zeta} \, U \xi = (P_{H_0 \xi, \eta} + P_{U \xi, \eta}) \zeta$ for all $\zeta \in \MH$.
    
    Taking the closure in the relation $K \restr \MD = (L + U) \restr \MD$ and using the conclusion \eqref{eq:closureRestrictedGenerator} of the first step, we find that
    \begin{equation*}
        \oln{(L + U) \restr \MD} = K \, .
    \end{equation*}
    Since $K$ is self-adjoint, this shows that the operator $(L + U) \restr \MD$ is essentially self-adjoint; therefore, $L + U$ is also essentially self-adjoint, and we have
    \begin{equation*}
        \oln{L + U} = \oln{(L + U) \restr \MD} = K \, .
    \end{equation*}
    
    This proves that the generator $K$ of $(W(t))_{t \in \R}$ is equal to $\oln{L + U}$, hence, by Stone's theorem (\cref{thm:Stone}), the unitary group $\ee^{\ii t (L + U)}$ must agree with the group $t \mapsto W(t)$, that is,
    \begin{equation*}
        \ee^{\ii t (L + U)} X = \ee^{\ii t H} X \, \ee^{- \ii t H_0} \quad \text{for all} \quad t \in \R \, , \ X \in \MK \, . \tag*{\qedhere}
    \end{equation*}
\end{proof}

We can now show that the two-sided Bogoliubov inequality \eqref{eq:BogoliubovVN} in von Neumann algebras reduces to the quantum-mechanical inequality \eqref{eq:BogoliubovQM} if $\MFM = \BO(\MH)$.

\begin{proposition}\label{pro:BogoliubovVNimpliesQM}
    Consider the von Neumann algebra $\MFM = \BO(\MH)$, let $H_0$ be a lower semi-bounded self-adjoint linear operator acting in $\MH$, let $U$ be relatively $H_0$-bounded such that $H = H_0 + U$ is self-adjoint, let $\beta > 0$, and assume that $\ee^{- \beta H_0}$ and $\ee^{- \beta H}$ are trace-class as well as $\norm{\ee^{- \beta V / 2} \, \ee^{- \beta H_0 / 2}}_\mathrm{HS} < + \infty$ and $\EV_\rho [V] < + \infty$.\footnote{The first assumption is satisfied, in particular, if $V$ is lower semi-bounded.} Then \cref{eq:BogoliubovVN} implies \cref{eq:BogoliubovQM}, i.e.,
    \begin{equation*}
        \EV_\rho[U] \le \Delta F \le \EV_{\rho_0}[U] \, .
    \end{equation*}
\end{proposition}

\begin{proof}
    Let $\omega \in \NS{\MFM}$ be given by the vector $\Omega \ce \rho_0^{1/2} = Z_0^{-1/2} \ee^{- \beta H_0 / 2}$, $Z_0 = \tr(\ee^{- \beta H_0})$, from the natural positive cone $\NPC = \HS(\MH)_+$, that is,
    \begin{equation*}
        \omega(A) = \frac{1}{Z_0} \tr \bigl(A \, \ee^{- \beta H_0}\bigr) = \EV_{\rho_0} [A] \com A \in \MFM \, .
    \end{equation*}
    It holds that $\omega$ is a $(\tau, \beta)$-KMS state, where $\tau_t(A) \ce \ee^{\ii t H_0} A \, \ee^{- \ii t H_0}$ for $t \in \R$ and $A \in \MFM$; see \cref{exa:LiouvillianKMS} \ref{enu:finiteQuantumSystem}. Moreover, $\omega$ is faithful because the operator $\ee^{- \beta H_0 /2}$ is injective.
    
    Let $L$ be the standard Liouvillian of $\tau$, cf. \cref{eq:LiouvillianFQS}. In the proof of \cref{lem:perturbedLiouvillianFQS}, we showed that $L + U$ is essentially self-adjoint, hence, Assumption \ref{enu:A1} is satisfied; employing an analogous argument, one can also show that \ref{enu:A2} holds true. Finally, according to the hypothesis of the proposition, we also have $\norm{\ee^{- \beta V / 2} \Omega} < + \infty$, that is, Assumption \ref{enu:A3}. Therefore, we can apply \cref{pro:unboundedPerturbedVector} to conclude that
    \begin{equation*}
        \Omega \in \dom (\ee^{- \beta (L + U) / 2}) \, .
    \end{equation*}

    We can now invoke \cref{thm:analyticCont}, implying that the function $\ii t \mapsto \ee^{\ii t (L + U)} \Omega$ has a continuous extension to the strip $\set{\alpha \in \C \, : \, - \beta / 2 \le \Re(\alpha) \le 0}$ which is analytic in the interior. Thus, we obtain from \cref{eq:perturbedLiouvillianFQS} of \cref{lem:perturbedLiouvillianFQS} that
    \begin{equation*}
        \Omega_{- \beta U} = \ee^{- \beta (L + U) / 2} \Omega = \ee^{- \beta H / 2} \Omega \, \ee^{\beta H_0 / 2} = Z_0^{-1/2} \ee^{- \beta H} \, .
    \end{equation*}
    The Hilbert-Schmidt norm of this vector is given by $\norm{\Omega_{- \beta U}}_\mathrm{HS} = (Z / Z_0)^{1/2}$, where $Z = \tr (\ee^{- \beta H})$, and thus the perturbed state $\omega_U$ defined in \cref{eq:unboundedPerturbedVector} takes the form
    \begin{equation*}
        \omega_U (A) = \frac{1}{\norm{\Omega_{- \beta U}}^2} \, \braket[\big]{\Omega_{- \beta U}, A \Omega_{- \beta U}} = \frac{1}{Z} \tr \bigl(A \, \ee^{- \beta H}\bigr) = \EV_\rho [A] \com A \in \MFM \, .
    \end{equation*}

    According to the final hypothesis of the proposition, also Assumption \ref{enu:A4} is satisfied. Therefore, we can use \cref{thm:BogoliubovVN} which implies that $\omega_U (U) \le \MF (\omega_U, \omega) \le \omega(U)$. Since $\omega_U$ is given by $\rho$ and $\omega$ is given by $\rho_0$, it follows that
    \begin{equation*}
        \MS_{\BO(\MH)} (\omega_U, \omega) = S (\rho, \rho_0) \, ,
    \end{equation*}
    cf.~\cref{exa:relativeEntropy}. Therefore, the relative free energy $\MF (\omega_U, \omega)$ from \cref{eq:relativeFreeEnergy} reduces to the quantity $\Delta F$ from \cref{eq:interfaceEnergy}, and thus the previous inequality takes the form of the quantum-mechanical two-sided Bogoliubov inequality \eqref{eq:BogoliubovQM}.
\end{proof}

\subsection{Application to a Bosonic scalar field}\label{subsec:bosonicField}

As an example, we evaluate the two-sided Bogoliubov inequality for a free Bose field whose modular generator is perturbed by conjugation with a Weyl unitary. The perturbation is inner, so the perturbed state is a coherent transform of the reference state; it is isospectral, the free energy vanishes, and both bounds collapse onto the symplectic form.

Let $M$ be the spacetime and $\mathcal{O} \subset M$ be a region. Consider the von Neumann algebra generated by the Weyl operators, $\MFM \ce \set[\big]{\ee^{\ii \ph(f)} : f \in C_c^\infty(\mathcal{O})}^{\prime\prime}$, where $C_c^\infty(\mathcal{O})$ denotes the real-valued smooth, compactly supported test functions on $\mathcal{O}$, satisfying the canonical commutation relations
\begin{equation}\label{eq:commutationRel}
  [\, \ph(f), \ph(h) \,]=  \ii \, \Gc(f,h) \, ,
\end{equation}
where $\Gc$ is the causal propagator (real, antisymmetric). Let $\Omega$ be cyclic and separating, with modular operator $\Delta_\Omega = \ee^{-\beta L}$, modular conjugation $J_\Omega$, and Tomita operator $S_\Omega = J_\Omega \Delta_\Omega^{1/2}$.  

\begin{lemma}\label{lem:innerPerturbation}
    Let $g$ be defined by $[\, \ii \ph(f), L \,] = \ph(g)$. Then conjugation of $L$ by the Weyl unitary $W: = \ee^{\ii \ph(f)}\in \MFM$ is exact, i.e.,
    \begin{equation*}
      W L W^\ast = L + U \quad \text{for} \quad U = \ph(g) - \frac{1}{2} \, \Gc(f,g) \, .
    \end{equation*}
\end{lemma}

\begin{proof}
    Let $X \ce \ii \ph(f)$ and $\operatorname{ad}_X (\,\cdot\,) \ce [X,\,\cdot\,]$. By the Baker-Campbell-Hausdorff (Hadamard) formula, we have
    \begin{equation*}
          W L W^\ast = \ee^{X} L \, \ee^{-X} = \sum_{n=0}^{\infty} \frac{1}{n!} \operatorname{ad}_X^{n} (L) = L + [X,L] + \frac{1}{2} \bigl[X, [X,L]\bigr] + \dotsb \, .
    \end{equation*}
    By definition of $g$, the first commutator is $[X,L] = [\,\ii \ph(f), L \,] = \ph(g)$. Moreover, the second commutator is just a number because \cref{eq:commutationRel} implies
    \begin{equation*}
        \frac{1}{2} \, \bigl[ X, [X,L] \bigr] = \frac{1}{2} \, [\, \ii \ph(f), \ph(g) \,] = - \frac{1}{2} \, \Gc(f,g) \, .
    \end{equation*}
    Therefore, all higher-order commutators vanish, $\operatorname{ad}_X^{n} (L) = 0$ for $n\ge3$, and the series terminates:
    \begin{equation*}
        W L W^\ast = L + \ph(g) - \frac{1}{2} \, \Gc(f,g) = L + U \, . \qedhere
    \end{equation*}
    \end{proof}
    
\begin{lemma}\label{lem:coherent}
    The perturbed vector $\Omega_{- \beta U}$ is a {two-sided coherent vector}:
    \begin{equation*}
        \Omega_{- \beta U} = W W^\prime \Omega \quad \text{with} \quad W^\prime = J_\Omega W J_\Omega \in \comm{\MFM} \, .
    \end{equation*}
    Furthermore, for every $A \in \MFM$, the perturbed state $\omega_U$ is given by
    \begin{equation*}
        \omega_U(A) = \omega (W^\ast A W) \, .
    \end{equation*}
    In particular, $\omega_U$ is unitarily equivalent to $\omega$, and their relative free energy \eqref{eq:relativeFreeEnergy} vanishes:
    \begin{equation*}
        \MF (\omega_U, \omega) = 0 \, .
    \end{equation*}
\end{lemma}

\begin{proof}
    Since $L + U = W L W^\ast$ by \cref{lem:innerPerturbation}, we also have $\ee^{- \beta (L + U) / 2} = W \ee^{- \beta L / 2} W^\ast$ by functional calculus \cite[Prop. 3.60]{Moretti19}. Using the defining property and polar decomposition of the Tomita operator $S_\Omega$ from \cref{eq:TomitaOperator,eq:TomitaPolarDecomp}, we can write
    \begin{equation*}
        \Delta_\Omega^{1/2} W^\ast \Omega = J_\Omega S_\Omega (W^\ast \Omega) = J_\Omega W \Omega \, .
    \end{equation*}
    Moreover, $J_\Omega W \Omega = J_\Omega W J_\Omega \Omega \equiv W^\prime \Omega$ according to \cref{eq:modularData}, where $W^\prime \ce J_\Omega W J_\Omega \in \comm{\MFM}$ by \cref{thm:Tomita}. Therefore, the perturbed vector \eqref{eq:unboundedPerturbedVector} is given by
    \begin{equation*}
        \Omega_{- \beta U} = \ee^{- \beta (L + U) / 2} \Omega = W \ee^{- \beta L / 2} W^\ast \Omega = W \Delta_\Omega^{1/2} W^\ast \Omega = W J_\Omega W J_\Omega \Omega = W W^\prime \Omega \, .
    \end{equation*}
    Since $W$ and $W^\prime$ are unitary operators, it follows that $\Omega_{- \beta U}$ is normalized, $\norm{\Omega_{- \beta U}} = 1$; hence, $\MF (\omega_U, \omega) = 0$ by virtue of \cref{eq:perturbedVectorRFE}. Finally, as $W^\prime \in \comm{\MFM}$ commutes with $A,W \in \MFM$:
    \begin{equation*}
        \omega_U(A) = \braket{\Omega_{- \beta U}, A \Omega_{- \beta U}} = \braket{W W^\prime \Omega, \, A \, W W^\prime \Omega} = \braket{\Omega, \, W^\ast A W \, \Omega} = \omega(W^\ast A W) \, . \qedhere
    \end{equation*}
\end{proof}
    
\begin{lemma}\label{lem:expectationValues}
    With $g = Lf$, we have
    \begin{equation*}
      \omega(U) = - \frac{1}{2} \, \Gc(f,g) \quad \text{and} \quad \omega_U(U) = + \frac{1}{2} \, \Gc(f,g).
    \end{equation*}
\end{lemma}

\begin{proof}
    Using the expression for $U$ from \cref{lem:innerPerturbation} and noting that the one-point function $\omega(\ph(g))$ vanishes, we immediately get $\omega(U) = - \frac{1}{2} \, \Gc(f,g)$. For the perturbed state, \cref{lem:coherent} and
    $W^\ast \ph(g) W = \ph(g) + [\, - \ii \ph(f), \ph(g) \,] = \ph(g) + \Gc(f,g)$ give
    \begin{equation*}
        \omega_U(U) = \omega(W^\ast U W) = \omega(\ph(g)) + \Gc(f,g) - \frac{1}{2} \, \Gc(f,g) = + \frac{1}{2} \, \Gc(f,g) \, . \qedhere
    \end{equation*}
\end{proof}

Substituting the results of \cref{lem:coherent,lem:expectationValues} into the two-sided Bogoliubov inequality $\omega_U(U) \le \MF (\omega_U, \omega) \le \omega(U)$, one obtains
\begin{equation*}
    \frac{1}{2} \, \Gc(f,g) \, \le \, 0 \, \le \, - \frac{1}{2} \, \Gc(f,g) \quad \Leftrightarrow \quad \Gc(f,Lf) \le 0 \, .
\end{equation*}
For a free field, the coherent-state relative entropy is symmetric and given by~\cite{Casini2019, Longo2019}
\begin{equation*}
    \MS_\MFM (\omega, \omega_U) = \MS_\MFM (\omega_U, \omega) = - \frac{\beta}{2} \, \Gc(f,g) \, .
\end{equation*}
Therefore, the two-sided Bogoliubov inequality is precisely the positivity of the relative entropy,
\begin{equation*}
    \MS_\MFM (\omega_U, \omega) = - \frac{1}{2} \, \Gc(f,Lf) \ge 0 \, ,
\end{equation*}
and its symmetrized (gap) form reads
\begin{equation*}
    \MS_\MFM (\omega_U, \omega) + \MS_\MFM (\omega, \omega_U) = - \beta \, \Gc(f,Lf) = \beta \, \big[\omega(U)-\omega_U(U)\big] \ge 0 \, .
\end{equation*}

In geometric settings, this relative-entropy gap acquires a horizon interpretation: through modular theory it is tied to the entropy--area
relation for quantum matter near spherically symmetric outer trapping horizons~\cite{Kurpicz2021}, to the entropy--area law and temperature of de Sitter horizons~\cite{DAngelo2024}, and to the semiclassical Einstein equations derived from the relative entropy~\cite{Dorau2026}.

\subsection{Variational Bounds}\label{subsec:variationalBounds}

Now, the variational bounds for the quantum-mechanical relative free energy stated in \cref{thm:variationalBoundsQM} shall be extended to the setting of von Neumann algebras. For bounded perturbations, they follow from the so-called Gibbs variational principle and the Donsker-Varadhan variational principle in von Neumann algebras that were established by D.~Petz \cite{Petz1988} (see \cref{pro:PetzVP} below). To extend \cref{thm:variationalBoundsQM}, we will generalize the two aforementioned variational principles to unbounded perturbations.

The following proposition contains the Gibbs variational principle for unbounded perturbations. The proof is an adaptation of Petz' argument, cf. Ref.~\cite[Prop. 1 \& Cor. 2]{Petz1988}, and relies on \cref{lem:unboundedPerturbedRelModOp4}. As before, let $(\MFM, \MH, J, \NPC)$ be a von Neumann algebra in standard form, let $\Phi, \Psi \in \NPC$ be separating, and define $L \ce - \inv{\beta} \log(\Delta_\Phi)$. Introduce the following final class of perturbations:
\begin{equation*}
    \MFS_3(L, \Phi, \Psi) \ce \left\{V \in \MFS_1(L, \Phi) \ : \ \parbox{0.45\textwidth}{\centering $\log(\Delta_{\Phi, \Psi}) - \beta V$ is essentially self-adjoint on $\dom(\log \Delta_{\Phi, \Psi}) \cap \dom(V)$}\right\} \, .
\end{equation*}

\begin{proposition}[Gibbs variational principle]\label{pro:GibbsVariationalPrinciple}
    Let $\varphi, \psi \in \NS{\MFM}$ be faithful with vector representatives $\Phi, \Psi \in \NPC$, define $L \ce - \inv{\beta} \log(\Delta_\Phi)$, and let $V \in \MFS_3(L, \Phi, \Psi)$. Then
    \begin{equation}\label{eq:generalizedPBinequality}
        - \inv{\beta} \log \bigl(\norm{\Phi_{- \beta V}}^2\bigr) \le \psi(V) + \inv{\beta} \MS_\MFM(\psi, \varphi) \, .
    \end{equation}
    Equality holds true if and only if $\psi = \omega_{\Phi_{- \beta V}} / \norm{\Phi_{- \beta V}}^2$. Furthermore,
    \begin{equation}\label{eq:generalizedGibbsVP}
        - \inv{\beta} \log\bigl(\norm{\Phi_{- \beta V}}^2\bigr) = \inf_{\gamma \in \Xi(V, \Phi)} \Bigl(\gamma(V) + \inv{\beta} \MS_\MFM(\gamma, \varphi)\Bigr) \, ,
    \end{equation}
    where the infimum is taken over the following set of normal states:
    \begin{equation*}
        \Xi(V, \Phi) \ce \left\{\omega \in \NS{\MFM} \ : \ \parbox{0.45\textwidth}{\centering $\omega = \omega_\Omega$ faithful and $\log(\Delta_{\Phi, \Omega}) - \beta V$ essentially self-adjoint}\right\} \, .
    \end{equation*}
\end{proposition}

\begin{proof}
    To establish \cref{eq:generalizedPBinequality}, we can assume, without loss of generality, that $\psi(V) < + \infty$; otherwise, since $\MS_\MFM(\psi, \varphi) \ge 0$ by \cref{thm:relEntProperties} \ref{enu:relEntNonNeg}, there is nothing to show. Define\footnote{The non-normalized functional $\varphi^V$ should not be confused with the \emph{state} $\omega_V$ defined in \cref{eq:unboundedPerturbedVector}.}
    \begin{equation*}
        \varphi^V(A) \ce \omega_{\Phi_{-\beta V}}(A) = \braket{\Phi_{-\beta V}, A \Phi_{-\beta V}} \com A \in \MFM \, ,
    \end{equation*}
    and recall that $\Phi_{- \beta V}$ is cyclic and separating for $\MFM$ (\cref{lem:perturbedVectorSeparating}). Therefore, we can consider the relative entropy from \cref{eq:relativeEntropy} between $\psi$ and $\varphi^V$ which is given by
    \begin{equation*}
        \MS_\MFM \bigl(\psi, \varphi^V\bigr) = - \braket[\big]{\Psi, \log (\Delta_{\Phi_V, \Psi}) \Psi} \, .
    \end{equation*}
    By our choice of the perturbation $V$, the operator $\log(\Delta_{\Phi, \Psi}) - \beta V$ is essentially self-adjoint; hence, we can apply \cref{lem:unboundedPerturbedRelModOp4} to write
    \begin{equation}\label{eq:relEntPerturbation}
        \MS_\MFM \bigl(\psi, \varphi^V\bigr) = - \braket[\big]{\Psi, \bigl(\log(\Delta_{\Phi, \Psi}) - \beta V\bigr) \Psi} = \MS_\MFM(\psi, \varphi) + \beta \psi(V) \, .
    \end{equation}
    Note that this identity is well-known for bounded perturbations \cite[Thm. 3.10]{Araki1977}. Using the fundamental inequality for $\MS_\MFM$ from \cref{thm:relEntProperties} \ref{enu:relEntLowerBound}, one obtains
    \begin{equation}\label{eq:basicLowerBound}
        \MS_\MFM \bigl(\psi, \varphi^V\bigr) \ge - \psi(\id_\MH) \log \left(\frac{\varphi^V(\ssupp_\psi)}{\psi(\id_\MH)}\right) = - \log \bigl(\varphi^V(\id_\MH)\bigr) \, ,
    \end{equation}
    where $\psi(\id_\MH) = 1$ was used. Since $\varphi^V(\id_\MH) = \norm{\Phi_{- \beta V}}^2$, the relation \eqref{eq:relEntPerturbation} together with the previous inequality \eqref{eq:basicLowerBound} implies that
    \begin{equation*}
        - \log \bigl(\norm{\Phi_{- \beta V}}^2\bigr) \le \beta \psi(V) + \MS_\MFM(\psi, \varphi) \, .
    \end{equation*}

    Having proved \cref{eq:generalizedPBinequality}, we will now examine when equality is reached in the latter. Observe first that, by virtue of \cref{thm:relEntProperties} \ref{enu:relEntScaling}, the lower bound in \cref{eq:basicLowerBound} can be interpreted as the relative entropy $\MS_{\MFM_0} \bigl(\psi|_{\MFM_0}, \varphi^V|_{\MFM_0}\bigr)$ on the subalgebra $\MFM_0 \ce \set{\id_\MH}$: since $\psi|_{\MFM_0} = \psi(\id_\MH) \, \1$ and $\varphi^V|_{\MFM_0} = \varphi^V(\id_\MH) \, \1$, where $\1 : A \mapsto 1$ denotes the constant functional, we have
    \begin{equation*}
        \MS_{\MFM_0} \bigl(\psi|_{\MFM_0}, \varphi^V|_{\MFM_0}\bigr) = \psi(\id_\MH) \, \mathord{\underbrace{\MS_{\MFM_0}(\1, \1)}_{=0}} - \psi(\id_\MH)^2 \log \left(\frac{\varphi^V(\id_\MH)}{\psi(\id_\MH)}\right) = - \log \bigl(\varphi^V(\id_\MH)\bigr) \, .
    \end{equation*}
    Therefore, \cref{eq:basicLowerBound} is just a special instance of monotonicity of the relative entropy when restricting the involved functionals to a subalgebra, cf. \cref{thm:relEntProperties} \ref{enu:relEntSubalg}:
    \begin{equation*}
        \MS_{\MFM_0} \bigl(\psi|_{\MFM_0}, \varphi^V|_{\MFM_0}\bigr) \le \MS_\MFM \bigl(\psi, \varphi^V\bigr) \, .
    \end{equation*}
    According to a well-known theorem of Petz \cite[Thm. 4]{Petz1986b} (see also Ref.~\cite[Thm. 9.3]{OP2004}), we have $\MS_{\MFM_0} \bigl(\psi|_{\MFM_0}, \varphi^V|_{\MFM_0}\bigr) = \MS_\MFM \bigl(\psi, \varphi^V\bigr)$ if and only if for all $t \in \R$,
    \begin{equation*}
        \Delta_{\varphi^V, \psi}^{\ii t} \, \Delta_{\psi}^{- \ii t} = \Delta_{\varphi^V|_{\MFM_0}, \psi|_{\MFM_0}}^{\ii t} \, \Delta_{\psi|_{\MFM_0}}^{- \ii t} \, .
    \end{equation*}
    Since $\psi|_{\MFM_0}$ is constant, $\Delta_{\psi|_{\MFM_0}} = \id_\MH$ \cite[Exa. 2.4]{Hiai2021}. Similarly, \cref{lem:relModOpScaling} implies that $\Delta_{\varphi^V|_{\MFM_0}, \psi|_{\MFM_0}} = \varphi^V(\id_\MH) / \psi(\id_\MH)$. Thus, we have equaliy in \cref{eq:basicLowerBound} if and only if
    \begin{equation*}
        \Delta_{\varphi^V, \psi}^{\ii t} \, \Delta_{\psi}^{- \ii t} = \left(\frac{\varphi^V(\id_\MH)}{\psi(\id_\MH)}\right)^{\ii t} \id_\MH
    \end{equation*}
    for all $t \in \R$. As Petz notes \cite[Prop. 1, p. 346]{Petz1988}, this identity implies that $\varphi^V = \lambda \psi$ for some $\lambda > 0$, that is, $\lambda = \varphi^V(\id_\MH) / \psi(\id_\MH) = \varphi^V(\id_\MH)$. We conclude that equality in \cref{eq:generalizedPBinequality} is attained only for the state
    \begin{equation*}
        \psi = \frac{1}{\varphi^V(\id_\MH)} \, \varphi^V = \frac{1}{\norm{\Phi_{- \beta V}}^2} \, \omega_{\Phi_{- \beta V}} \, .
    \end{equation*}
    
    Let us finally prove the Gibbs variational principle \eqref{eq:generalizedGibbsVP}. On the one hand, the previously established inequality \eqref{eq:generalizedPBinequality} implies that
    \begin{equation*}
        - \inv{\beta} \log \bigl(\norm{\Phi_{- \beta V}}^2\bigr) \le \inf_{\gamma \in \Xi(V, \Phi)} \Bigl(\gamma(V) + \inv{\beta} \MS_\MFM(\gamma, \varphi)\Bigr)
    \end{equation*}
    because, by assumption on $\psi$ and $V$, it holds that $\psi \in \Xi(V, \Phi)$. Consider the perturbation $W \ce V + \inv{\beta} \log (\norm{\Phi_{- \beta V}}^2) \, \id_\MH \in \MFS_1(L, \Phi)$ and recall from the proof of \cref{lem:unboundedPerturbationRelEnt2} that
    \begin{equation*}
        \Phi_{- \beta W} = \frac{\Phi_{- \beta V}}{\norm{\Phi_{- \beta V}}} \, .
    \end{equation*}
    According to \cref{lem:unboundedPerturbationRelModOp2}, there holds
    \begin{align*}
        \log \bigl(\Delta_{\Phi, \Phi_{- \beta W}}\bigr) - \beta V &= \log(\Delta_\Phi) + \beta J W J - \beta V \\
        &= \log(\Delta_\Phi) + \beta J V J + \log \bigl(\norm{\Phi_{- \beta V}}^2\bigr) \, \id_\MH - \beta V \\
        &= - \beta (L + V - J V J) + \log \bigl(\norm{\Phi_{- \beta V}}^2\bigr) \, \id_\MH \\
        &= - \beta L_V + \log \bigl(\norm{\Phi_{- \beta V}}^2\bigr) \, \id_\MH \, .
    \end{align*}
    Since $V \in \MFS_1(L, \Phi)$ satisfies Assumption \ref{enu:A2}, it follows that $\log \bigl(\Delta_{\Phi, \Phi_{- \beta W}}\bigr) - \beta V$ is essentially self-adjoint; hence, the state $\wt{\psi} \ce \omega_{\Phi_{- \beta W}}$ is an element of $\Xi(V, \Phi)$. Therefore, as the choice $\psi = \wt{\psi}$ yields equality in \cref{eq:generalizedPBinequality}, we also obtain the reverse inequality
    \begin{equation*}
        - \inv{\beta} \log \bigl(\norm{\Phi_{- \beta V}}^2\bigr) = \wt{\psi}(V) + \inv{\beta} \MS_\MFM\bigl(\wt{\psi}, \varphi\bigr) \ge \inf_{\gamma \in \Xi(V, \Phi)} \Bigl(\gamma(V) + \inv{\beta} \MS_\MFM(\gamma, \varphi)\Bigr) \, . \tag*{\qedhere}
    \end{equation*}
\end{proof}

\begin{remark}
    In Donald's approach \cite{Donald1990} to perturbation theory mentioned above (\cref{subsec:unboundedPerturbations}), given $\varphi \in \NS{\MFM}$ and an extended-valued, semi-bounded operator $V$ affiliated with $\MFM$, one shows that the Gibbs variational principle \eqref{eq:generalizedGibbsVP} has a unique minimizer which is then called $\varphi^V$ \cite[Thm. 3.1]{Donald1990}. In contrast, here we showed that, defining $\varphi^V$ for a general unbounded $V$ via \cref{pro:unboundedPerturbedVector}, the Gibbs variational principle is satisfied.
\end{remark}

The next goal is to obtain a dual form of the Gibbs variational principle, expressing the Araki-Uhlmann relative entropy as a supremum over unbounded perturbations; an expression of this form is usually called Donsker-Varadhan principle \cite[Lem. 2.1]{Donsker1975}. In the setting of von Neumann algebras, Petz \cite[Thm. 9]{Petz1988} proved such a variational principle for bounded perturbations.

\begin{proposition}[Petz]\label{pro:PetzVP}
    Let $\varphi, \psi \in \NS{\MFM}$ be faithful with $\varphi$ given by $\Phi \in \NPC$. Then the Araki-Uhlmann relative entropy of $\psi$ and $\varphi$ can be written as
    \begin{align}\label{eq:PetzVP}
        \MS_\MFM(\psi, \varphi) = \sup_{V \in \MFM^\mathrm{sa}} \Bigl(- \beta \psi(V) - \log \bigl(\norm{\Phi_{- \beta V}}^2\bigr)\Bigr) \, .
    \end{align}
\end{proposition}

A proof of this can also be found in Ref.~\cite[p. 228]{OP2004}. In the following, we will use the Gibbs variational principle and the previous result of Petz to obtain an expression similar to \cref{eq:PetzVP} for unbounded perturbations.

\begin{proposition}[Donsker-Varadhan principle]\label{pro:unboundedDVVP}
    Let $(\MFM, \MH, J, \NPC)$ be a von Neumann algebra in standard form, let $\varphi, \psi \in \NS{\MFM}$ be faithful with vector representatives $\Phi, \Psi \in \NPC$, and let $L = - \inv{\beta} \log(\Delta_\Phi)$. Then
    \begin{equation}\label{eq:unboundedDVVP}
        \MS_\MFM(\psi, \varphi) = \sup_{V \in \MFS_3(L, \Phi, \Psi)} \Bigl(- \beta \psi(V) - \log \bigl(\norm{\Phi_{- \beta V}}^2\bigr)\Bigr) \, .
    \end{equation}
\end{proposition}

\begin{proof}
    From the inequality \eqref{eq:generalizedPBinequality} in \cref{pro:GibbsVariationalPrinciple}, it immediately follows that
    \begin{equation*}
        \MS_\MFM(\psi, \varphi) \ge \sup_{V \in \MFS_3(L, \Phi, \Psi)} \Bigl(- \beta \psi(V) - \log \bigl(\norm{\Phi_{- \beta V}}^2\bigr)\Bigr) \, .
    \end{equation*}
    On the other hand, since the set $\MFM^\mathrm{sa}$ of self-adjoint elements of $\MFM$ satisfies $\MFM^\mathrm{sa} \subset \MFS_3(L, \Phi, \Psi)$, Petz' variational expression \eqref{eq:PetzVP} for bounded perturbations shows that
    \begin{equation*}
        \MS_\MFM(\psi, \varphi) \le \sup_{V \in \MFS_3(L, \Phi, \Psi)} \Bigl(- \beta \psi(V) - \log \bigl(\norm{\Phi_{- \beta V}}^2\bigr)\Bigr) \, . \tag*{\qedhere}
    \end{equation*}
\end{proof}

Combining the Gibbs variational principle from \cref{pro:GibbsVariationalPrinciple} and the Donsker-Varadhan principle from the previous \cref{pro:unboundedDVVP}, we finally obtain the desired variational bounds for the relative free energy.

\begin{theorem}[Variational bounds]\label{thm:variationalBoundsVN}
    Let $(\MFM, \MH, J, \NPC)$ be a von Neumann algebra in standard form, let $\tau$ be a $W^\ast$-dynamics on $\MFM$ with standard Liouvillian $L$, and let $\omega$ be a faithful $(\tau, \beta)$-KMS state ($\beta > 0$) given by $\Omega \in \NPC$. For all $U \in \MFS_2(L, \Omega)$, it holds that
    \begingroup
    \small
    \begin{equation*}
        \inf_{\psi \in \Xi(U, \Omega)} \Bigl(\psi(U) + \inv{\beta} \MS_\MFM(\psi, \omega)\Bigr) = \MF(\omega_U, \omega) = \sup_{V \in \MFS_3(L, \Omega, \Omega_{- \beta U})} \Bigl(\omega_U(U - V) - \inv{\beta} \log \bigl(\norm{\Omega_{- \beta V}}^2\bigr)\Bigr) \, .
    \end{equation*}
    \endgroup
    In particular, we have the following family of two-sided bounds for the relative free energy $\MF(\omega_U, \omega)$, which holds for all $V \in \MFS_3(L, \Omega, \Omega_{- \beta U})$ and $\psi \in \Xi(U, \Omega)$:
    \begin{equation}\label{eq:unboundedFamilyBounds}
        \omega_U(U - V) - \inv{\beta} \log\bigl(\norm{\Omega_{- \beta V}}^2\bigr) \le \MF(\omega_U, \omega) \le \psi(U) + \inv{\beta} \MS_\MFM(\psi, \omega) \, .
    \end{equation}
    Setting $V = 0$ and $\psi = \omega$ in the previous inequalities, one recovers the two-sided Bogoliubov inequality \eqref{eq:BogoliubovVN}.
\end{theorem}

\begin{proof}
    Choosing $\varphi = \omega$ in the Gibbs variational principle \eqref{eq:generalizedGibbsVP} and recalling that $\MF(\omega_U, \omega) = - \inv{\beta} \log (\norm{\Omega_{- \beta U}}^2)$ by virtue of \cref{eq:perturbedVectorRFE}, we obtain the first equality. Combining the definition \eqref{eq:relativeFreeEnergy} of the relative free energy $\MF(\omega_U, \omega) = \omega_U(U) + \inv{\beta} \MS_\MFM(\omega_U, \omega)$ with the Donsker-Varadhan variational principle \eqref{eq:unboundedDVVP} gives the second equality:
    \begin{align*}
        \MF(\omega_U, \omega) &= \omega_U(U) + \sup_{V \in \MFS_3(L, \Omega, \Omega_U)} \Bigl(- \omega_U(V) - \inv{\beta} \log \bigl(\norm{\Omega_{- \beta V}}^2\bigr)\Bigr) \\
        &= \sup_{V \in \MFS_3(L, \Omega, \Omega_U)} \Bigl(\omega_U(U - V) - \inv{\beta} \log \bigl(\norm{\Omega_{- \beta V}}^2\bigr)\Bigr) \, . \tag*{\qedhere}
    \end{align*}
\end{proof}

\section{Conclusions and outlook}

We have generalized the two-sided Bogoliubov inequality \cite{Reible2022} to the setting of von Neumann algebras, utilizing the framework of unbounded perturbation theory of KMS states from Ref.~\cite{DJP03} as well as some new extensions of it. Furthermore, we have also generalized the variational bounds for the relative free energy \cite{Reible2022} by extending the known variational principles for the relative entropy in terms of bounded perturbations \cite{Petz1988} to the unbounded case. The extension of the results of Ref.~\cite{Reible2022} to general von Neumann algebras has been motivated by the idea to establish a thermodynamic criterion for the quantification of entanglement in terms of the free energy of separation, which is applicable not only to finite quantum systems but also to systems with infinitely many degrees of freedom, like quantum fields of spins, as they play an increasingly important role in the development of quantum technology.

While the current contribution was devoted to establishing the criterion on a mathematically solid basis, in future work we intend to apply the inequality to concrete systems in order to illustrate its utility. On the mathematical side, it would also be interesting to investigate the question whether the domains of the infimum and supremum in \cref{thm:variationalBoundsVN} can be enlarged, that is, whether the essential self-adjointness condition related to the relative modular operator can be replaced or dropped. Finally, on a more foundational level, we note that there still remain some conceptual difficulties about the relativistic covariance of thermodynamic properties like the KMS condition or thermalization; see Ref.~\cite{Passegger2025} for a detailed discussion of this topic. Therefore, the physical interpretation of the quantities entering the two-sided Bogoliubov inequality needs to be handled with some care in the context of relativistic quantum field theories.

\appendix
\section{Existence of perturbed vectors}\label{app:perturbedVectors}

In the following, we provide a proof of the key \cref{pro:unboundedPerturbedVector}, asserting that if $(\MFM, \MH, J, \NPC)$ is a von Neumann algebra in standard form, $\Omega \in \NPC$ a normalized and separating vector, $L \ce - \inv{\beta} \log(\Delta_\Omega)$ with $\beta > 0$, and $V \in \MFS_1(L, \Omega)$, then
\begin{equation*}
    \Omega \in \dom \bigl(\ee^{- \beta (L + V) / 2}\bigr) \tand \ee^{- \beta (L + V) / 2} \Omega = \dlim{w}{n \to \infty} \ee^{- \beta (L + V_n) / 2} \Omega \, .
\end{equation*}
The proof follows exactly the cunning strategy that was devised in Ref.~\cite[Thm. 5.1 (1)]{DJP03}, and below we will show that our \cref{pro:unboundedPerturbedVector} is actually equivalent to the result stated in Ref.~\cite[Thm. 5.5 (1)]{DJP03}.

\begin{proof}[Proof of \cref{pro:unboundedPerturbedVector}]
    According to Assumption \ref{enu:A3}, we have $\int_\R \ee^{- \beta \lambda} \diff \mu_\Omega^V(\lambda) < + \infty$, where the measure $\mu_\Omega^V$ is defined by $\mu_\Omega^V (A) \ce \braket{\Omega, E_V(A) \Omega}$ for all Borel sets $A \subset \R$, with $E_V$ is the spectral measure of $V$. For every $n \in \N$, let $V_n = f_n(V) \in \MFM$ be defined as in \cref{eq:approximatingSequence}. Since
    \begin{equation*}
        \ee^{- \beta f_n(\lambda)} \le 1 + \ee^{- \beta \lambda} \ (\lambda \in \R) \quad \text{and} \quad \int_\R (1 + \ee^{- \beta \lambda}) \diff \mu_\Omega^V(\lambda) = \norm{\Omega}^2 + \norm{\ee^{-\beta V / 2} \Omega}^2 < + \infty \, ,
    \end{equation*}
    the dominated convergence theorem implies that
    \begin{align*}
        \lim_{n \to \infty} \norm[\big]{\ee^{- \beta V_n / 2} \Omega}^2 = \int_\R \lim_{n \to \infty} \ee^{- \beta f_n(\lambda)} \diff \mu_\Omega^V(\lambda) = \int_\R \ee^{- \beta \lambda} \diff \mu_\Omega^V(\lambda) = \norm[\big]{\ee^{- \beta V / 2} \Omega}^2 \, .
    \end{align*}   
    In particular, this shows that there exists $C > 0$ such that for all $n \in \N$, $\norm{\ee^{- \beta V_n / 2} \Omega} \le C$. Using the Golden-Thompson inequality (\cref{pro:GoldenThompson}), one therefore obtains
    \begin{equation}\label{eq:auxiliaryExistenceInequality}
        \norm[\big]{\ee^{- \beta (L + V_n) / 2} \Omega} = \norm[\big]{\ee^{(\log \Delta_\Omega - \beta  V_n) / 2} \Omega} \overset{\eqref{eq:GoldenThompson}}{\le} \norm[\big]{\ee^{- \beta  V_n / 2} \Omega} \le C \, .
    \end{equation}

    Recall that $L + V_n \to L + V$ in the strong resolvent sense according to \cref{lem:convergenceLiouvillian} \ref{enu:convUnperturbedL}. Since $L + V_n$ is self-adjoint and $L + V$ is essentially self-adjoint by \ref{enu:A1}, we can apply \cref{pro:SRimpliesFSR} to conclude that also $\ee^{- \beta (L + V_n) / 2} \to \ee^{- \beta (L + V) / 2}$ in the strong resolvent sense. Now, using \cref{pro:SRimpliesW} (with $T_n \equiv \ee^{- \beta (L + V_n) / 2}$, $T \equiv \ee^{- \beta (L + V) / 2}$, and $\Omega_n \equiv \Omega$ for all $n \in \N$, noting that \cref{eq:auxiliaryExistenceInequality} yields the necessary bound), we can conclude that
    \begin{equation*}
        \Omega \in \dom(\ee^{- \beta (L + V) / 2}) \, ,
    \end{equation*}
    which is the first assertion, and we also directly obtain the statement about weak convergence:
    \begin{equation*}
        \dlim{w}{n \to \infty} \ee^{- \beta (L + V_n) / 2} \Omega = \ee^{- \beta (L + V) / 2} \Omega \, . \tag*{\qedhere}
    \end{equation*}
\end{proof}

Let us now show that \cref{pro:unboundedPerturbedVector} is indeed equivalent to the perturbation result of Ref.~\cite[Thm. 5.5 (1)]{DJP03}.

\begin{proposition}\label{pro:equivalencePerturbedVectors}
    Let $(\MFM, \MH, J, \NPC)$ be a von Neumann algebra in standard form and let $\beta > 0$. The following assertions are equivalent:
    \begin{enumerate}[equiv]
        \item If $\Phi \in \NPC$ is normalized and separating for $\MFM$, $L_\Phi \ce - \inv{\beta} \log(\Delta_\Phi)$, and $V \in \MFS_1(L_\Phi, \Phi)$, then $\Phi \in \dom \bigl(\ee^{- \beta (L_\Phi + V) / 2}\bigr)$.
        
        \item If $\tau$ is a $W^\ast$-dynamics on $\MFM$ with standard Liouvillian $L$, $\omega = \omega_\Omega \in \NS{\MFM}$ is a faithful $(\tau, \beta)$-KMS state, and $V \in \MFS_1(L, \Omega)$, then $\Omega \in \dom \bigl(\ee^{- \beta (L + V) / 2}\bigr)$.
    \end{enumerate}
\end{proposition}

\begin{proof}
    (i) $\implies$ (ii): Since $\omega = \omega_\Omega$ is a faithful state, it follows that $\Omega \in \NPC$ is normalized and separating. Moreover, \cref{pro:LiouvillianKMS} shows that the Liouvillian $L$ of $\tau$ is given by $L = - \inv{\beta} \log (\Delta_\Omega)$. Thus, (i) implies $\Omega \in \dom \bigl(\ee^{- \beta (L + V) / 2}\bigr)$.

    (ii) $\implies$ (i): According to Takesaki's fundamental theorem, $\varphi = \omega_\Phi$ is a faithful $(\sigma_t^\Phi, -1)$-KMS state; cf. \cref{exa:LiouvillianKMS} \ref{enu:modularGroup}. Hence, $\varphi$ is also a $(\sigma_{- \inv{\beta} t}^\Phi, \beta)$-KMS state \cite[p. 78]{BR2}. The standard Liouvillian of the $W^\ast$-dynamics $t \mapsto \sigma_{- \inv{\beta} t}^\Phi$ is given by $L = - \inv{\beta} \log(\Delta_\Phi) = L_\Phi$, and so (ii) implies that $\Phi \in \dom \bigl(\ee^{- \beta (L_\Phi + V) / 2}\bigr)$.
\end{proof}

\section{Strong resolvent convergence}\label{app:SRconvergence}

Let $(T_n)_{n \in \N}$ and $T$ be self-adjoint linear operators acting in a Hilbert space $\MH$. The sequence $(T_n)_{n \in \N}$ is said to converge to the operator $T$ in the \emph{strong resolvent sense} iff $(T_n - \ii)^{-1} \to (T - \ii)^{-1}$ in the strong operator topology, that is,
\begin{equation*}
    \forall \, \xi \in \MH \ : \ \lim_{n \to \infty} \norm[\big]{(T_n - \ii)^{-1} \xi - (T - \ii)^{-1} \xi} = 0 \, .
\end{equation*}
In the following, we collect some propositions characterizing strong resolvent convergence from the literature, which are needed in the main part of the paper.

The first result, which is proved, e.g., in Refs.~\cite[Prop. 10.1.18]{Oliveira2009} and \cite[Thm. VIII.25]{RS1}, shows that strong convergence on a core implies strong resolvent convergence.

\begin{proposition}\label{pro:coreSOimpliesSR}
    Let $T$ and $(T_n)_{n \in \N}$ be self-adjoint operators acting in $\MH$. Suppose that $\MD \subset \MH$ is a subspace contained in $\dom(T)$ and $\dom(T_n)$ for all $n \in \N$, and that $T$ is essentially self-adjoint on $\MD$. Furthermore, assume that $\lim_{n \to \infty} T_n \xi = T \xi$ for all $\xi \in \MD$. Then $T_n \to T$ in the strong resolvent sense.
\end{proposition}

The next result, taken from and proved in Ref.~\cite[Prop. A.4]{DJP03}, is used to construct a weakly convergent sequence of vectors from operators converging in the strong resolvent sense.

\begin{proposition}\label{pro:SRimpliesW}
    Let $T$ and $(T_n)_{n \in \N}$ be self-adjoint operators acting in $\MH$ such that $T_n \to T$ in the strong resolvent sense. Furthermore, suppose that $(\Omega_n)_{n \in \N} \subset \MH$ and $\Omega \in \MH$ are vectors such that $\Omega_n \to \Omega$ weakly, and that there exists a constant $C \ge 0$ such that for all $n \in \N$, $\norm{T_n \Omega_n} \le C$. Then $\Omega \in \dom(T)$, $\dlim{w}{n \to \infty} T_n \Omega_n$ exists, and $T \Omega = \dlim{w}{n \to \infty} T_n \Omega_n$.
\end{proposition}

The following well-known result characterizes strong resolvent convergence in terms of strong convergence of bounded functions of the operators \cite[Prop. 10.1.9]{Oliveira2009}, \cite[Thm. VIII.20]{RS1}.

\begin{proposition}\label{pro:characterizationSRConvergence}
    $T_n \to T$ in the strong resolvent sense if and only if $f(T_n) \to f(T)$ in the strong operator topology for all bounded continuous $f : \R \to \C$.
\end{proposition}

The final result is an immediate consequence of the previous one.

\begin{proposition}\label{pro:SRimpliesFSR}
    Let $(T_n)_{n \in \N}$ be a sequence of self-adjoint operators acting in a Hilbert space $\MH$ and converging in the strong resolvent sense to a self-adjoint operator $T$. Then for any continuous real-valued function $g : \R \to \R$, it holds that $g(T_n) \to g(T)$ in the strong resolvent sense as $n \to + \infty$.
\end{proposition}

\begin{proof}
    Let $z \in \C \setminus \R$ be arbitrary, and observe that the function $h_z : \R \to \C$, $\lambda \mapsto \bigl(g(\lambda) - z\bigr)^{-1}$, is continuous and bounded. Indeed, since clearly $\abs{g(\lambda) - z} \ge \abs{\Im(z)}$, one obtains
    \begin{equation*}
        \abs{h_z(\lambda)} \le \frac{1}{\abs{\Im(z)}} \, .
    \end{equation*}
    Therefore, since $T_n \to T$ in the strong resolvent sense, it follows from \cref{pro:characterizationSRConvergence} that $h_z(T_n) \to h_z(T)$ as $n \to + \infty$ in the strong operator topology. But $h_\ii(T_n) = (g(T_n) - \ii)^{-1}$ and $h_\ii(T) = (g(T) - \ii)^{-1}$, hence the previous statement is equivalent to convergence $g(T_n) \to g(T)$ in the strong resolvent sense.
\end{proof}

\section{One-parameter groups of operators}\label{app:oneParameterGroups}

In this appendix, we gather some well-known results about strongly continuous one-parameter unitary groups that are used at various places in the main text. Recall that a \emph{strongly continuous one-parameter unitary group} on a Hilbert space $\MH$ is a family $(U(t))_{t \in \R}$ of unitary operators $U(t) : \MH \to \MH$ satisfying $U(t) U(s) = U(t + s)$ for all $t, s \in \R$ and $\lim_{h \to 0} U(t + h) \xi = U(t) \xi$ for all $\xi \in \MH$ and $t \in \R$.

The following assertion shows how to recover a self-adjoint operator from the strongly continuous unitary group it generates \cite[Prop. 6.1]{Schmuedgen2012}.

\begin{proposition}\label{pro:generatorUnitaryGroup}
    Let $\MH$ be a Hilbert space and $T$ be a self-adjoint linear operator in $\MH$. Define $U(t) \ce \ee^{\ii t T}$ for all $t \in \R$. Then $(U(t))_{t \in \R}$ is a strongly continuous one-parameter unitary group on $\MH$, and the operator $T$ is uniquely determined by the group as follows:
    \begin{gather*}
        \dom(T) = \Set{\xi \in \MH \ : \ \odbound{}{t}{t=0} U(t) \xi = \lim_{h \to 0} \frac{1}{h} \, \bigl(U(t) \xi - \xi\bigr) \in \MH \ \text{exists}} \, , \\
        T \xi = - \ii \, \odbound{}{t}{t=0} U(t) \xi \com \xi \in \dom(T) \, .
    \end{gather*}
    Moreover, the domain of $T$ is invariant under the unitary group, that is,
    \begin{equation*}
        \forall \, \xi \in \dom(T) \ \forall \, t \in \R \ : \ U(t) \xi \in \dom(T) \, .
    \end{equation*}
\end{proposition}

A converse to the previous result is the classical theorem of Stone, which states that every unitary group is of the form discussed in \cref{pro:generatorUnitaryGroup}; see, e.g., Ref.~\cite[Thm. 6.2]{Schmuedgen2012}.

\begin{theorem}[Stone]\label{thm:Stone}
    If $(U(t))_{t \in \R}$ is a strongly continuous one-parameter unitary group on $\MH$, then there exists a uniquely defined self-adjoint linear operator $T$ such that $U(t) = \ee^{\ii t T}$ for all $t \in \R$.
\end{theorem}

The next result provides a criterion to identify cores of generators of strongly continuous unitary groups \cite[Prop. II.1.7]{EngelNagel2000}, \cite[Prop. 6.3]{Schmuedgen2012}.

\begin{proposition}\label{pro:coreGenerator}
    Let $T$ be the generator of a strongly continuous one-parameter unitary group $(U(t))_{t \in \R}$ on $\MH$. If $\MD \subset \dom(T)$ is a subspace which is dense in $\MH$ and satisfies $U(t) \MD \subset \MD$ for all $t \in \R$, then $\MD$ is a core for $A$.
\end{proposition}

The final result that we need is an analytic continuation theorem due to G.~K.~Pedersen and M.~Takesaki \cite[Lem. 3.2]{Pedersen1973}; see also Refs.~\cite[Lem. 3]{Araki1973b} and \cite[Cor. 9.21]{StratilaZsido2019}.

\begin{theorem}\label{thm:analyticCont}
    Let $T$ be a positive self-adjoint operator with $\ker(T) = \set{0}$, $\xi \in \MH$, and $s < 0$. Then $\xi \in \dom(T^s)$ if and only if the function $\ii t \mapsto T^{\ii t} \xi \in \MH$ can be extended to a continuous function in $\set{\alpha \in \C \, : \, s \le \Re(\alpha) \le 0}$ which is analytic in $\set{\alpha \in \C \, : \, s < \Re(\alpha) < 0}$.
\end{theorem}

\section*{Acknowledgments}

The material in this article is based on the first-named author's master thesis \cite{Reible2025mt}, where a detailed presentation of modular theory and perturbation theory can be found. B.M.R.~also wishes to thank Professor Jan Derezi\'{n}ski for providing helpful comments about his 2003 \emph{Reviews in Mathematical Physics} paper \cite{DJP03}.

\bibliographystyle{abbrv}
\bibliography{bogoliubov.bib}

\begin{thebibliography}{10}

\bibitem{Amico2008}
L.~Amico, R.~Fazio, A.~Osterloh, and V.~Vedral.
\newblock Entanglement in many-body systems.
\newblock {\em Rev. Mod. Phys.}, 80(2):517--576, 2008.

\bibitem{Araki1973a}
H.~Araki.
\newblock Expansional in {B}anach algebras.
\newblock {\em Ann. Sci. de l'École Norm. Sup.}, 4(6):67--84, 1973.

\bibitem{Araki1973b}
H.~Araki.
\newblock {G}olden-{T}hompson and {P}eierls-{B}ogolubov inequalities for a general von {N}eumann algebra.
\newblock {\em Commun. Math. Phys.}, 34:167–178, 1973.

\bibitem{Araki1973c}
H.~Araki.
\newblock Relative {H}amiltonian for faithful normal states of a von {N}eumann algebra.
\newblock {\em Publ. Res. Inst. Math. Sci.}, 9(1):165–209, 1973.

\bibitem{Araki1974}
H.~Araki.
\newblock Some properties of modular conjugation operator of von {N}eumann algebras and a non-commutative {R}adon-{N}ikodym theorem with a chain rule.
\newblock {\em Pacific J. Math.}, 50(2):309--354, 1974.

\bibitem{Araki1976}
H.~Araki.
\newblock Relative entropy of states of von {N}eumann algebras.
\newblock {\em Publ. Res. Inst. Math. Sci.}, 11(3):809--833, 1976.

\bibitem{Araki1977}
H.~Araki.
\newblock Relative entropy of states of von {N}eumann algebras {II}.
\newblock {\em Publ. Res. Inst. Math. Sci.}, 13(1):173--192, 1977.

\bibitem{Attal2006}
S.~Attal.
\newblock Elements of operator algebras and modular theory.
\newblock In S.~Attal, A.~Joye, and C.-A. Pillet, editors, {\em Open Quantum Systems I. The Hamiltonian Approach}, Lecture Notes in Mathematics 1880, chapter~3, pages 69--105. Springer, Berlin, 2006.

\bibitem{Bakshi2024}
A.~Bakshi, A.~Liu, A.~Moitra, and E.~Tang.
\newblock High-temperature {G}ibbs states are unentangled and efficiently preparable.
\newblock In {\em 2024 IEEE 65th Annual Symposium on Foundations of Computer Science (FOCS)}, pages 1027--1036, Chicago, IL, USA, 2024.
\newblock arXiv:2403.16850 [quant-ph].

\bibitem{BlanchardBruening2015}
P.~Blanchard and E.~Brüning.
\newblock {\em Mathematical Methods in Physics}.
\newblock Progress in Mathematical Physics 69. Birkhäuser, Cham, second edition, 2015.

\bibitem{BR1}
O.~Bratteli and D.~W. Robinson.
\newblock {\em Operator Algebras and Quantum Statistical Mechanics I. $C^\ast$- and $W^\ast$-Algebras, Symmetry Groups. Decomposition of States}.
\newblock Theoretical and Mathematical Physics. Springer, Berlin, 1979.

\bibitem{BR2}
O.~Bratteli and D.~W. Robinson.
\newblock {\em Operator Algebras and Quantum Statistical Mechanics II. Equilibrium States, Models in Quantum Statistical Mechanics}.
\newblock Theoretical and Mathematical Physics. Springer, Berlin, 1981.

\bibitem{BreiteneckerGruemm1972}
M.~Breitenecker and H.-R. Grümm.
\newblock Note on trace inequalities.
\newblock {\em Commun. Math. Phys.}, 26:276–279, 1972.

\bibitem{Casini2019}
H.~Casini, S.~Grillo, and D.~Pontello.
\newblock Relative entropy for coherent states from {A}raki formula.
\newblock {\em Phys. Rev. D}, 99(12):125020, 2019.

\bibitem{Cole2025}
M.~C. Cole, M.~T. Pelly, C.~V. Topping, T.~Reindl, U.~Waizmann, J.~Weis, and A.~W. Rost.
\newblock Coulomb blockade thermometry based nanocalorimetry.
\newblock {\em Rev. Sci. Instrum.}, 96(7):073903, 2025.

\bibitem{Connes1974}
A.~Connes.
\newblock Caractérisation des espaces vectoriels ordonnés sous-jacents aux algèbres de von {N}eumann.
\newblock {\em Ann. Inst. Fourier}, 24(4):121--155, 1974.

\bibitem{DAngelo2024}
E.~D'Angelo, M.~B. Fr\"ob, S.~Galanda, P.~Meda, A.~Much, and K.~Papadopoulos.
\newblock Entropy-area law and temperature of de {S}itter horizons from modular theory.
\newblock {\em Prog. Theor. Exp. Phys.}, 2024(2):021A01, 2024.

\bibitem{Oliveira2009}
C.~R. de~Oliveira.
\newblock {\em Intermediate Spectral Theory and Quantum Dynamics}.
\newblock Progress in Mathematical Physics 54. Birkhäuser, Basel, 2009.

\bibitem{DelleSite2015}
L.~Delle~Site.
\newblock Shannon entropy and many-electron correlations: Theoretical concepts, numerical results, and {C}ollins conjecture.
\newblock {\em Int. J. Quantum Chem.}, 115:1396--1404, 2015.

\bibitem{DelleSite2017}
L.~Delle~Site, G.~Ciccotti, and C.~Hartmann.
\newblock Partitioning a macroscopic system into independent subsystems.
\newblock {\em J. Stat. Mech.}, 2017(8):083201, 2017.

\bibitem{DelleSite2024pra}
L.~Delle~Site and C.~Hartmann.
\newblock Scaling law for the size dependence of a finite-range quantum gas.
\newblock {\em Phys. Rev. A}, 109(2):022209, 2024.

\bibitem{DJP03}
J.~Derezi\'{n}ski, V.~Jak\v{s}i\'{c}, and C.-A. Pillet.
\newblock Perturbation theory of {$W^\ast$}-dynamics, {L}iouvilleans and {KMS}-states.
\newblock {\em Rev. Math. Phys.}, 15(5):447--489, 2003.

\bibitem{Donald1987}
M.~J. Donald.
\newblock Free energy and the relative entropy.
\newblock {\em J. Stat. Phys.}, 49(1-2):81--87, 1987.

\bibitem{Donald1990}
M.~J. Donald.
\newblock Relative {H}amiltonians which are not bounded from above.
\newblock {\em J. Funct. Anal.}, 91(1):143--173, 1990.

\bibitem{Donsker1975}
M.~D. Donsker and S.~R.~S. Varadhan.
\newblock Asymptotic evaluation of certain {M}arkov process expectations for large time, {I}.
\newblock {\em Commun. Pure Appl. Math.}, 28(1):1--47, 1975.

\bibitem{Dorau2026}
P.~Dorau and A.~Much.
\newblock From quantum relative entropy to the semiclassical {E}instein equations.
\newblock {\em Phys. Rev. Lett.}, 136(9):091602, 2026.

\bibitem{EngelNagel2000}
K.-J. Engel and R.~Nagel.
\newblock {\em One-Parameter Semigroups for Linear Evolution Equations}.
\newblock Graduate Texts in Mathematics 194. Springer, New York, 2000.

\bibitem{Feng2024}
C.~Feng and L.~Chen.
\newblock Quantifying quantum entanglement via machine learning models.
\newblock {\em Commun. Theor. Phys.}, 76:075104, 2024.

\bibitem{Haag1967}
R.~Haag, N.~M. Hugenholtz, and M.~Winnink.
\newblock On the equilibrium states in quantum statistical mechanics.
\newblock {\em Commun. Math. Phys.}, 5:215–236, 1967.

\bibitem{Haagerup1975}
U.~Haagerup.
\newblock The standard form of von {N}eumann algebras.
\newblock {\em Math. Scand.}, 37(2):271--283, 1975.

\bibitem{Hiai2021}
F.~Hiai.
\newblock {\em Lectures on Selected Topics in von Neumann Algebras}.
\newblock EMS Series of Lectures in Mathematics. European Mathematical Society, 2021.

\bibitem{Hiai2021mps}
F.~Hiai.
\newblock {\em Quantum $f$-Divergences in von Neumann Algebras}.
\newblock Mathematical Physics Studies. Springer, Singapore, 2021.

\bibitem{Horodecki2009}
R.~Horodecki, P.~Horodecki, M.~Horodecki, and K.~Horodecki.
\newblock Quantum entanglement.
\newblock {\em Rev. Mod. Phys.}, 81(2):865--942, 2009.

\bibitem{Keyl2003}
M.~Keyl, D.~Schlingemann, and R.~F. Werner.
\newblock Infinitely entangled states.
\newblock {\em Quant. Inf. Comput.}, 3(4):281--306, 2003.

\bibitem{Kinoshita2025}
S.~Kinoshita, K.~Murata, D.~Yamamoto, and R.~Yoshii.
\newblock Spin systems as quantum simulators of quantum field theories in curved spacetimes.
\newblock {\em Phys. Rev. Res.}, 7(2):023917, 2025.

\bibitem{Koutny2023}
D.~Koutn\'y, L.~Gin\'es, L.~Mocza{\l}a-Dusanowska, S.~H\"{o}fling, C.~Schneider, A.~Predojevi\'c, and M.~Je\v{z}ek.
\newblock Deep learning of quantum entanglement from incomplete measurements.
\newblock {\em Sci. Adv.}, 9:eadd7131, 2023.

\bibitem{Kurpicz2021}
F.~Kurpicz, N.~Pinamonti, and R.~Verch.
\newblock Temperature and entropy-area relation of quantum matter near spherically symmetric outer trapping horizons.
\newblock {\em Lett. Math. Phys.}, 111:110, 2021.

\bibitem{Laurell2025}
P.~Laurell, A.~Scheie, E.~Dagotto, and D.~A. Tennant.
\newblock Witnessing entanglement and quantum correlations in condensed matter: A review.
\newblock {\em Adv. Quantum Technol.}, 8:2400196, 2025.

\bibitem{Li2019}
D.~Li, K.~Lee, B.~Y. Wang, M.~Osada, S.~Crossley, H.~R. Lee, Y.~Cui, Y.~Hikita, and H.~Y. Hwang.
\newblock Superconductivity in an infinite-layer nickelate.
\newblock {\em Nature}, 572:624--627, 2019.

\bibitem{Lindblad73}
G.~Lindblad.
\newblock Entropy, information and quantum measurement.
\newblock {\em Commun. Math. Phys.}, 33:305--322, 1973.

\bibitem{Longo1997}
R.~Longo.
\newblock An analogue of the {K}ac--{W}akimoto formula and black hole conditional entropy.
\newblock {\em Commun. Math. Phys.}, 186:451–479, 1997.

\bibitem{Longo2019}
R.~Longo.
\newblock Entropy of coherent excitations.
\newblock {\em Lett. Math. Phys.}, 109(12):2587–2600, 2019.

\bibitem{Moretti19}
V.~Moretti.
\newblock {\em Fundamental Mathematical Structures of Quantum Theory}.
\newblock Springer, Cham, 2019.

\bibitem{OP2004}
M.~Ohya and D.~Petz.
\newblock {\em Quantum Entropy and Its Use}.
\newblock Theoretical and Mathematical Physics. Springer, Berlin, corrected second printing edition, 2004.

\bibitem{Passegger2025}
A.~G. Passegger and R.~Verch.
\newblock Disjointness of inertial {KMS} states and the role of {L}orentz symmetry in thermalization.
\newblock {\em Rev. Math. Phys.}, 37(2):2430009, 2025.

\bibitem{Pedersen1973}
G.~K. Pedersen and M.~Takesaki.
\newblock The {R}adon-{N}ikodym theorem for von {N}eumann algebras.
\newblock {\em Acta Math.}, 130:53--87, 1973.

\bibitem{Petz1986a}
D.~Petz.
\newblock Properties of the relative entropy of states of von {N}eumann algebras.
\newblock {\em Acta Math. Hung.}, 47(1-2):65--72, 1986.

\bibitem{Petz1986b}
D.~Petz.
\newblock Sufficient subalgebras and the relative entropy of states of a von {N}eumann algebra.
\newblock {\em Commun. Math. Phys.}, 105:123–131, 1986.

\bibitem{Petz1988}
D.~Petz.
\newblock A variational expression for the relative entropy.
\newblock {\em Commun. Math. Phys.}, 114:345–349, 1988.

\bibitem{Pillet2006}
C.-A. Pillet.
\newblock Quantum dynamical systems.
\newblock In S.~Attal, A.~Joye, and C.-A. Pillet, editors, {\em Open Quantum Systems I. The Hamiltonian Approach}, Lecture Notes in Mathematics 1880, chapter~4, pages 107--182. Springer, Berlin, 2006.

\bibitem{Plenio2007}
M.~B. Plenio and S.~Virmani.
\newblock An introduction to entanglement measures.
\newblock {\em Quant. Inf. Comput.}, 7(1):1--51, 2007.

\bibitem{Popescu1997}
S.~Popescu and D.~Rohrlich.
\newblock Thermodynamics and the measure of entanglement.
\newblock {\em Phys. Rev. A}, 56(5):R3319(R), 1997.

\bibitem{RS1}
M.~Reed and B.~Simon.
\newblock {\em Methods of Modern Mathematical Physics. Volume I: Functional Analysis}.
\newblock Academic Press, San Diego, {R}evised and {E}nlarged edition, 1980.

\bibitem{Reible2025mt}
B.~M. {Reible}.
\newblock Monotonicity of the relative entropy and the two-sided {B}ogoliubov inequality in von {N}eumann algebras.
\newblock Master thesis, Universität Leipzig, arXiv:2501.04564 [math.OA], 2024.

\bibitem{Reible2022}
B.~M. Reible, C.~Hartmann, and L.~{Delle Site}.
\newblock Two-sided {Bogoliubov} inequality to estimate finite size effects in quantum molecular simulations.
\newblock {\em Lett. Math. Phys.}, 112:97, 2022.

\bibitem{Reible2025apx}
B.~M. {Reible}, C.~Hartmann, and L.~{Delle Site}.
\newblock Finite-size effects in molecular simulations: a physico-mathematical view.
\newblock {\em Adv. Phys. X}, 10(1):2495151, 2025.

\bibitem{Reible2023}
B.~M. Reible, J.~F. Hille, C.~Hartmann, and L.~{Delle Site}.
\newblock Finite-size effects and thermodynamic accuracy in many-particle systems.
\newblock {\em Phys. Rev. Res.}, 5(2):023156, 2023.

\bibitem{Reible2025camcos}
B.~M. {Reible}, N.~Liebreich, C.~Hartmann, and L.~{Delle Site}.
\newblock A fast and rigorous numerical tool to measure length-scale artifacts in molecular simulations.
\newblock Accepted for publication in \textit{Commun. Appl. Math. Comp. Sci.}, in press. arXiv:2511.01442 [physics.comp-ph], 2025.

\bibitem{Rubboli2024}
R.~Rubboli and M.~Tomamichel.
\newblock New additivity properties of the relative entropy of entanglement and its generalizations.
\newblock {\em Commun. Math. Phys.}, 405:162, 2024.

\bibitem{Schmuedgen2012}
K.~Schmüdgen.
\newblock {\em Unbounded Self-adjoint Operators on Hilbert Space}.
\newblock Graduate Texts in Mathematics 265. Springer, Dordrecht, 2012.

\bibitem{Stamatova2025}
M.~Stamatova and V.~Vedral.
\newblock Complex heat capacity as a witness of spatio-temporal entanglement.
\newblock arXiv:2508.15728 [quant-ph]. Accepted for publication in \emph{Phys.~Rev.~A}., 2026.

\bibitem{Stratila2020}
S.~V. Strătilă.
\newblock {\em Modular Theory in Operator Algebras}.
\newblock Cambridge IISc Series. Cambridge University Press, Cambridge, second edition, 2020.

\bibitem{StratilaZsido2019}
S.~V. Strătilă and L.~Zsidó.
\newblock {\em Lectures on von Neumann Algebras}.
\newblock Cambridge IISc Series. Cambridge University Press, Cambridge, second edition, 2019.

\bibitem{Takesaki1970}
M.~Takesaki.
\newblock {\em Tomita's Theory of Modular Hilbert Algebras and its Applications}.
\newblock Lecture Notes in Mathematics 128. Springer, Berlin, 1970.

\bibitem{Uhlmann1977}
A.~Uhlmann.
\newblock Relative entropy and the {W}igner-{Y}anase-{D}yson-{L}ieb concavity in an interpolation theory.
\newblock {\em Commun. Math. Phys.}, 54:21--32, 1977.

\bibitem{Umegaki62}
H.~Umegaki.
\newblock Conditional expectation in an operator algebra, {IV}. {E}ntropy and information.
\newblock {\em Kodai Math. Sem. Rep.}, 14(2):59--85, 1962.

\bibitem{vanLuijk2025cmp}
L.~van Luijk, A.~Stottmeister, R.~F. Werner, and H.~Wilming.
\newblock Pure state entanglement and von {N}eumann algebras.
\newblock {\em Commun. Math. Phys.}, 406:296, 2025.

\bibitem{Vedral2002}
V.~Vedral.
\newblock The role of relative entropy in quantum information theory.
\newblock {\em Rev. Mod. Phys.}, 74(1):197--234, 2002.

\bibitem{Vedral1998}
V.~Vedral and M.~B. Plenio.
\newblock Entanglement measures and purification procedures.
\newblock {\em Phys. Rev. A}, 57(3):1619--1633, 1998.

\bibitem{Vedral1997}
V.~Vedral, M.~B. Plenio, M.~A. Rippin, and P.~L. Knight.
\newblock Quantifying entanglement.
\newblock {\em Phys. Rev. Lett.}, 78(12):2275--2279, 1997.

\bibitem{Verch2005}
R.~Verch and R.~F. Werner.
\newblock Distillability and positivity of partial transposes in general quantum field systems.
\newblock {\em Rev. Math. Phys.}, 17(5):545--576, 2005.

\bibitem{Xie2023}
S.~Xie, Y.-Y. Zhao, C.~Zhang, Y.-F. Huang, C.-F. Li, G.-C. Guo, and J.~H. Eberly.
\newblock Experimental examination of entanglement estimates.
\newblock {\em Phys. Rev. Lett.}, 130(15):150801, 2023.

\end{thebibliography}

\end{document}